\documentclass[twocolumn,secnumarabic,amssymb, nobibnotes, aps, prd]{revtex4-1}

\usepackage{longtable,graphicx}        
\usepackage{color}        

\makeindex             
\usepackage{amsmath,eepic,amsthm}

\def\frak#1\mathfrak{#1}
\def\bbbc{{\mathbb{C}}}

\def\bbbr{{\Bbb R}}

\def\pvint{-{\kern-1.30em}\int}
\def\Ref#1#2{(\ref{#2})}

\def\openone{\leavevmode\hbox{\small1\kern-0.355em\normalsize1}}
\def\biglb{\big[\hspace*{-.7mm}\big[}
\def\bigrb{\big]\hspace*{-.7mm}\big]}
\def\Biglb{\Big[\hspace*{-1.4mm}\Big[}
\def\Bigrb{\Big]\hspace*{-1.3mm}\Big]}
\def\Res{\mathop{\mbox{Res}\,}\limits}

\def\Bigglb{\Bigg[\hspace*{-1.4mm}\Bigg[}
\def\Biggrb{\Bigg]\hspace*{-1.3mm}\Bigg]}

\def\bPhi{{\boldsymbol \Phi}}

\def\bP{{\boldsymbol P}}
\def\bQ{{\boldsymbol Q}}

\def\bPsi{{\boldsymbol \Psi}}

\def\bTheta{{\boldsymbol \Theta}}
\def\Res{\mathop{\mbox{Res}\,}\limits}
\def\wedgecomma{\mathop{\wedge}\limits_{'}}

\def\ad{\mbox{ad}\,}

\def\tr{\mbox{tr}\,}
\def\Im{\mbox{Im}\,}

\def\Re{\mbox{Re}\,}

\def\pd#1,#2{\frac{\partial#1}{\partial#2}}
\def\d#1,#2{\frac{d#1}{d#2}}
\def\lpd#1,#2{\frac{\stackrel{\rightarrow}{\partial#1}}{\partial#2}}
\def\rpd#1,#2{\frac{\stackrel{\leftarrow}{\partial#1}}{\partial#2}}
\def\col#1,#2,#3,#4 {\begin{array}{c}#1\\ #2\\ #3\\#4 \end{array}}
\newtheorem{theorem}{Theorem}
\newtheorem{remark}{Remark}
\newtheorem{lemma}{Lemma}
\newtheorem{Prop}{Proposition}
\newtheorem{corollary}{Corollary}
\newtheorem{example}{Example}
\allowdisplaybreaks
\numberwithin{equation}{section}

\begin{document}
\title[Complete integrability of NL NLS]{Complete integrability of Nonlocal Nonlinear Schr\"odinger equation}
\author{V. S. Gerdjikov$^1$, A. Saxena$^2$ \\
$^1$ Institute of Nuclear Research and Nuclear Energy,\\
Bulgarian Academy of Sciences,\\
72 Tsarigradsko chausee, Sofia 1784, Bulgaria\\
e-mail: gerjikov@inrne.bas.bg\\[5pt]
$^2$Los Alamos National Laboratory\\
Theoretical Division and Center for Nonlinear Studies \\
Los Alamos, NM 87545, USA\\
e-mail: avadh@lanl.gov
}

\begin{abstract}
Based on the completeness relation for the squared solutions of the Lax operator $L$ we show that a subset
of nonlocal equations from the hierarchy of nonlocal nonlinear Schr\"odinger equations (NLS)  is a completely
integrable system.  The spectral properties of the Lax operator indicate that there are two types of soliton
solutions.  The relevant action-angle variables are parametrized by the scattering data of the Lax operator.
The notion of the symplectic basis, which directly maps the variations of the potential of $L$
to the variations of the action-angle variables has been generalized to the nonlocal case.
We also show that the inverse scattering method can be viewed as a generalized Fourier transform.  
Using the trace identities and the symplectic basis we construct the  hierarchy Hamiltonian structures 
for the  nonlocal NLS equations.

\end{abstract}

\maketitle
\tableofcontents

\section{Introduction}

Nonlocal, nonlinear equations arise in a variety of physical contexts ranging from hydrodynamics to optics
to condensed matter and high energy physics.  The integrability of such equations is an important issue not
only for the physical properties and soliton solutions but also in its own right, see Refs. \cite{AbBak,AblMus,AKevFra,Mats,%
MorKon,PerKon,AblMus2,Yan} and the numerous references therein. We note that the nonlocal NLS and its associated symmetry reductions are
contained in \cite{AblMus} and its inverse scattering transform is given in \cite{AblMus2}.

In this context, first we shall consider the generic AKNS system \cite{AKNS}
\begin{equation}\label{eq:AKNS0}\begin{split}
L\psi  &\equiv i {d\psi   \over dx }  + U(x,t,\lambda ) \psi (x,t,\lambda ) =0, \\
U(x,t,\lambda ) &=  q(x,t) - \lambda \sigma _3, \qquad q(x,t) = \left(\begin{array}{cc} 0 & q_+ \\
q_- & 0  \end{array}\right).
\end{split}\end{equation}
The potential $q(x,t)$ belongs to the class of smooth functions vanishing fast enough for $x\to\pm\infty$. By generic
here we mean that the complex-valued functions $q_+(x,t)$ and $q_-(x,t)$ are independent.

One of the paradigms related to the generic AKNS system allows Lax representation with an $M$-operator of the form:
\begin{equation}\label{eq:AKNS4}\begin{split}
M\psi  &\equiv i {d\psi   \over dt }  + V(x,t,\lambda ) \psi (x,t,\lambda ) =0, \\
V(x,t,\lambda ) &=  V_0(x,t) + \lambda q(x,t) - \lambda^2 \sigma _3, \\
V_0(x,t) &= \frac{1}{2} q_+q_- \sigma_3 + \frac{i}{4} [\sigma_3, q_x].
\end{split}\end{equation}
The two operators $L$ and $M$ commute identically with respect to $\lambda$ provided $q_\pm (x,t)$ satisfy the
nonlinear evolution equations (NLEE):
\begin{equation}\label{eq:gnls}\begin{split}
i \frac{\partial q_+}{ \partial t } + \frac{1}{2} \frac{\partial^2 q_+}{ \partial x^2 } + q_+^2 q_-(x,t) &=0, \\
-i \frac{\partial q_-}{ \partial t } + \frac{1}{2} \frac{\partial^2 q_-}{ \partial x^2 } + q_+ q_-^2(x,t) &=0.
\end{split}\end{equation}

This system, which we will call a generalized nonlinear Schr\"odinger (NLS) equatio system, can be simplified by applying
one of the next reductions:
\begin{description}
  \item[A)]  $q_+(x,t) = \epsilon_1 q_-^*(x,t)$,  $\epsilon_1^2=1$
  \item[B)] $ q_+(x,t) = \epsilon_3 q_-(x,t)$,  $\epsilon_3^2=1$
  \item[C)] $ q_+(x,t) = \epsilon_2 q_-^*(-x,t)$,  $\epsilon_2^2=1$
\end{description}

The class A) NLEE  contains the famous NLS equation \cite{ZaSha} ($q_+=u(x,t) $)
\begin{equation}\label{eq:nls}\begin{split}
i \frac{\partial u}{ \partial t } + \frac{1}{2} \frac{\partial^2 u}{ \partial x^2 } + |u|^2 u(x,t) &=0,
\end{split}\end{equation}
whose integrability was discovered by Zakharov and Shabat \cite{ZaSha,ZMNP}.
The class B) reduction allowed AKNS to integrate  the sine-Gordon and the mKdV equations \cite{AKNS,FaTa}. Here
we just note that the $M$-operators for the sine-Gordon and mKdV equations have different $\lambda$-dependence.

The last  reduction C) has been used recently to analyze the class containing the nonlocal NLS eq.
\begin{equation}\label{eq:nls2}\begin{split}
i \frac{\partial u}{ \partial t } + \frac{1}{2} \frac{\partial^2 u}{ \partial x^2 } + u^2 (x,t) u^*(-x,t) &=0,
\end{split}\end{equation}
which recently became physically important \cite{AblMus,AblMus2}.

The paper is organized as follows.  Section 2 contains some well known facts about the spectral theory of $L$ such as:
a) the construction of the  fundamental analytic solutions (FAS);  b) the mappings between the potential $q(x,t)$ (resp. the variation of the potential $\delta q(x,t)$)
and the scattering data (resp. the variation of the scattering data) based on the Wronskian relations.
In Section 3 we prove the completeness relation for the `squared solutions' of $L$ and derive the expansions of $q(x,t)$ and
$\sigma_3 \delta q(x,t)$ over them. We also introduce the recursion operators for which the `squared solutions' are
complete sets of eigenfunctions. In Section IV we use the results of Section III to derive the fundamental properties of the
generic NLEE.  These include description of the class of NLEE related to $L$, their integrals of motion, trace identities and the
 action-angle variables for the generic NLEE related to $L$ (\ref{eq:AKNS0}).
In the next Section V we construct the hierarchy of Hamiltonian structures.   We demonstrate
that the symplectic basis of squared solutions is a very convenient tool in this analysis. It allows one to prove effectively the
compatibility of the hierarchy of Poisson brackets (subsection V.1) and the compatibility of the hierarchy of symplectic forms $\Omega_{(m)}^\mathbb{C}$.
The symplectic basis allows us to give a simple proof of the fact that all $\Omega_{(m)}^\mathbb{C}$ acquire canonical form in terms of the
action-angle variables. In Section VI we analyze the effects of the local (subsection VI.1) and non-local involutions (subsection VI.2) on the
Hamiltonian hierarchies and on the action-angle variables. In both cases we derive the action-angle variables in terms of the scattering data
of the corresponding reduced Lax operator.
In the case of local  involutions we reproduce the well known results \cite{AKNS,Vg-EKh1,Vg-EKh2,KaNew,GeYaV}. In the case of non-local involutions
the results are new and provide us with the action-angle variables for the nonlocal NLS. In the last Section VII we discuss the results and draw conclusions.

\section{Preliminaries }
\subsection{Direct and Inverse Scattering Problem for the Generic AKNS system }

We start by briefly outlining the main results from the direct and inverse scattering problems for the generic AKNS system
on the class of potentials $q(x,t)$ that are smooth functions vanishing fast enough for $x\to\pm\infty$. The basic tools in this theory
are the Jost solutions, which are introduced by:
\begin{equation}\label{eq:AKNS1}\begin{split}
\lim_{x\to -\infty} e^{-i \lambda \sigma_3 x}\phi(x,t,\lambda) & = \openone, \\
\lim_{x\to \infty} e^{-i \lambda \sigma_3 x}\psi(x,t,\lambda)  &= \openone,
\end{split}\end{equation}
and the scattering matrix
\begin{equation}\label{eq:AKNS2}\begin{split}
T(\lambda, t) \equiv \hat{\psi} (x,t,\lambda) \phi(x,t,\lambda) = \left(\begin{array}{cc} a^+ & -b^- \\
b^+ & a^- \end{array}\right) .
\end{split}\end{equation}

Time dependence of $T(\lambda,t)$ is given by
\begin{equation}\label{eq:AKNS5}\begin{split}
i \frac{\partial T}{ \partial t } - \lambda^2 [\sigma_3, T(\lambda,t)]=0 .
\end{split}\end{equation}

The analyticity properties of the Jost solutions and the FAS of $L$ are given as follows:
\begin{equation}\label{eq:Fas0}\begin{aligned}
\psi(x,t,\lambda) &= |\psi^-,  \psi^+|, & \quad \phi(x,t,\lambda) &= |\phi^+,  \phi^-|, \\
\chi^+(x,t,\lambda) &= |\phi^+,  \psi^+|, & \quad \chi^-(x,t,\lambda) &= |\psi^-,  \phi^-|, \\
\det \chi^\pm (x,t,\lambda) &= a^\pm(\lambda).
\end{aligned}\end{equation}
Here the superscript $+$ (resp. $-$) means that the corresponding  column or function
allows an analytic extension for $\lambda \in \mathbb{C}_+$ (resp. for $\lambda \in \mathbb{C}_-$).

In what follows we will also need  the asymptotics of $\chi^\pm (x,\lambda)$ for $x\to\pm\infty$. From (\ref{eq:AKNS2})
and .r(\ref{eq:Fas0}) we find
\begin{equation}\label{eq:limx0}\begin{aligned}
\lim_{x\to -\infty} e^{i\lambda \sigma_3 x} \chi^+(x,\lambda) & = S^+(\lambda), \\
\lim_{x\to -\infty} e^{i\lambda \sigma_3 x} \chi^-(x,\lambda) & = S^-(\lambda), \\
\lim_{x\to \infty} e^{i\lambda \sigma_3 x} \chi^+(x,\lambda) & = T^-(\lambda), \\
\lim_{x\to \infty} e^{i\lambda \sigma_3 x} \chi^-(x,\lambda) & = T^+(\lambda),
\end{aligned}\end{equation}
where
\begin{equation}\label{eq:STpm}\begin{aligned}
S^+(x,\lambda) &=  \left(\begin{array}{cc} 1 & b^- \\ 0 & a^+  \end{array}\right), &\;
S^-(x,\lambda) &=  \left(\begin{array}{cc} a^- & 0 \\ -b^+ & 1  \end{array}\right), \\
T^+(x,\lambda) &=  \left(\begin{array}{cc} 1 & -b^- \\ 0 & a^-  \end{array}\right), &\;
T^-(x,\lambda) &=  \left(\begin{array}{cc} a^+ & 0 \\ b^+ & 1  \end{array}\right) .
\end{aligned}\end{equation}
We will also need  the asymptotics of the FAS for $\lambda \to\infty$:
\begin{equation}\label{eq:chilim}\begin{split}
\lim_{\lambda\to\infty} \chi^\pm(x,\lambda) e^{i\lambda\sigma_3 x} = \openone.
\end{split}\end{equation}
As an immediate consequence of  (\ref{eq:chilim}) and the last line of eq. (\ref{eq:Fas0}) there follows:
\begin{equation}\label{eq:apmlim}\begin{split}
 \lim_{\lambda\to\infty}a^\pm(\lambda) =1.
\end{split}\end{equation}

\subsection{The Wronskian relations}\label{sec:l5-1}

The analysis of the mapping $\mathcal{ F} \colon \mathcal{ M} \to
\mathcal{ T} $ between the class of the allowed potentials
$\mathcal{ M} $ and the scattering data of $L $ starts with the so-called Wronskian relations. As we shall see,
they would allow us to
\begin{enumerate}

\item  formulate the idea that the ISM is a generalized Fourier transform (GFT);

\item  determine explicitly the proper generalizations of the usual exponents;

\item introduce the skew--scalar product  on $\mathcal{M} $ which endows it with a symplectic structure.

\end{enumerate}

All these ideas  will be worked out for the generic Zakharov-Shabat system (\ref{eq:AKNS0}).
Along with it, we will need also
\begin{equation}\label{eq:I.3}\begin{split}
i {d\hat{\psi} \over dx }- \hat{\psi} (x,t,\lambda )U(x,t,\lambda ) =0 ,
\end{split}\end{equation}
\begin{multline}\label{eq:I.4}
i {d \delta \psi   \over dx }  + \delta U(x,t,\lambda ) \psi
(x,t,\lambda ) \\ + U(x,t,\lambda ) \delta \psi (x,t,\lambda ) =0 ,
\end{multline}
\begin{multline}\label{eq:I.5}
 i {d\dot{\psi} \over dx }  + \dot{U}(x,t,\lambda) \psi (x,t,\lambda )\\
  + U(x,t,\lambda ) \dot{\psi} (x,t,\lambda ) =0 ,
\end{multline}
which one can associate with (\ref{eq:AKNS0}).
By $\hat{\psi} $ above we denote $\psi ^{-1} $;  $\delta \psi  $ is the variation of the Jost solution
which corresponds to a given variation $\delta q(x,t) $ of the potential;  by dot we denote the derivative
with respect to the spectral parameter; for example $\dot{U}(x,t,\lambda ) = -\sigma _3 $.

\subsection{The mapping from $q $ to $\mathcal{ T} $}\label{ssec:wr1}

The first Wronskian identity reads:
\begin{multline}\label{eq:wr.1}
\left. \left( \hat{\chi }\sigma _3 \chi (x,\lambda ) - \sigma _3\right) \right|_{-\infty }^{\infty } \\
=  - i\int_{-\infty }^{\infty } dx\, \hat{\chi }[q(x), \sigma_3 ] \chi (x,\lambda ),
\end{multline}
where $\chi (x,\lambda ) $ can be any fundamental
solution of $L $. For convenience we choose them to be the FAS introduced earlier.

Skipping the details we obtain the following relations between  the reflection coefficients and the potential:
\begin{equation}\label{eq:wr.12}\begin{aligned}
\rho^\pm(\lambda ) &\equiv  \frac{b^\pm}{a^\pm} = {i  \over (a^\pm (\lambda ))^2 } \biglb q(x), \bPhi^\pm(x,\lambda ) \bigrb , \\
\tau ^\pm(\lambda ) &\equiv  \frac{b^\mp}{a^\pm}  = {i  \over (a^\pm (\lambda ))^2 } \biglb q(x), \bPsi^\pm(x,\lambda ) \bigrb ,
\end{aligned}\end{equation}
where $\mathcal{ E}_+^+(x,\lambda )$ are the `squared solutions'
\begin{eqnarray}\label{eq:wr.6}
\mathcal{ E}_+^+(x,\lambda ) = \chi ^+(x,\lambda ) \sigma _+ \hat{\chi }^+(x,\lambda ) .
\end{eqnarray}
By $\biglb X , Y \bigrb$ we denote the skew-scalar product.
\begin{multline}\label{eq:wr.7}
\biglb X , Y \bigrb  \equiv {1  \over  2} \int_{-\infty }^{\infty }
dx\, \tr (X(x), [\sigma _3 , Y(x)]) \\  = - \biglb Y , X \bigrb .
\end{multline}
The ``squared'' solutions $\bPhi^\pm(x,\lambda ) $ and   $
\bPsi^\pm(x,\lambda ) $ are defined by:
\begin{multline}\label{eq:wr.14}
\bPhi^\pm(x,\lambda ) = a^\pm (\mathcal{ E}_\pm^\pm(x,\lambda
))^{\rm f}\\ = \left( \begin{array}{cc} 0 & \pm (\phi _1^\pm (x,\lambda ))^2 \\
\mp (\phi _2^\pm (x,\lambda ))^2 & 0\end{array} \right) ,
\end{multline}
\begin{multline}\label{eq:wr.15}
\bPsi^\pm(x,\lambda ) = a^\pm (\mathcal{E}_\mp^\pm(x,\lambda ))^{\rm f}\\
 = \left( \begin{array}{cc} 0 & \mp (\psi _1^\pm (x,\lambda ))^2 \\
\pm (\psi _2^\pm(x,\lambda ))^2 & 0\end{array} \right) .
\end{multline}
Here and below by $X^{\rm f}$ we mean the off-diagonal part of the matrix $X$.

These ``squared'' solutions effectively coincide with the ones
that appeared originally in \cite{AKNS,Kaup,KaNew,Vg-EKh1,Vg-EKh2}; the difference is that here we have used
the gauge covariant formulation proposed in \cite{GeYa}, see also \cite{GeYaV}.

\subsection{The mapping from $\delta q $ to $\delta \mathcal{ T}$}\label{ssec:wr2}

The second type of Wronskian relations, relates the variation of the potential $\delta q(x) $
to the corresponding variations of the scattering data. To
this purpose we start with the identity:

\begin{eqnarray}\label{eq:wr.18}
\left. \hat{\chi }\delta  \chi (x,\lambda ) \right|_{-\infty
}^{\infty } = i\int_{-\infty }^{\infty } dx\, \hat{\chi }\delta
q(x)\chi (x,\lambda ) ,
\end{eqnarray}
which gives:
\begin{eqnarray}\label{eq:wr.21}
\delta \rho ^\pm(\lambda ) &=&  {\mp i  \over 2(a^\pm(\lambda ))^2}
\biglb [\sigma _3 ,\delta q(x)], \bPhi^\pm (x,\lambda ) \bigrb , \\
\label{eq:wr.22} \delta \tau ^\pm(\lambda ) &=&  {\pm i  \over
2(a^\pm(\lambda ))^2 } \biglb [\sigma _3 ,\delta q(x)], \bPsi^\pm
(x,\lambda ) \bigrb ,
\end{eqnarray}
and
\begin{equation}\label{eq:dAlam}
\delta \mathcal{ A}(\lambda ) =- {i  \over 4a^\pm (\lambda ) }
\biglb [\sigma _3 ,\delta q(x)], \bTheta^\pm (x,\lambda ) \bigrb .
\end{equation}
Here $\bPsi^\pm(x,\lambda ) $ and $\bPhi^\pm(x,\lambda ) $ are the
same ``squared'' solutions as in (\ref{eq:wr.14}), (\ref{eq:wr.15})
\begin{multline}\label{eq:Theta}
\bTheta^\pm(x,\lambda ) = a^+(\lambda ) (\chi ^\pm (x,\lambda )
\sigma _3 \hat{\chi }^\pm(x,\lambda ))^{\rm f} \\
= \left( \begin{array}{cc} 0 & -2(\phi _1^\pm \psi_1^\pm)(x,\lambda ) \\
2(\phi _2^\pm \psi _2^\pm)(x,\lambda ) & 0 \end{array} \right),
\end{multline}
and $\mathcal{ A}(\lambda ) =\pm \ln a^\pm(\lambda) $ for $\lambda\in \mathbb{C}_\pm$.

These relations are basic in the analysis of the NLEE related to
the Zakharov-Shabat system and their Hamiltonian structures.
 We shall use them later assuming
that the variation of $q(x) $ is due to its time evolution. In
this case $q(x,t) $ depends on $t $ in such a way, that it
satisfies certain NLEE. Then we consider variations of the type:
\begin{eqnarray}\label{eq:wr.23}
\delta  q(x,t) = \pd{q},{t} \delta t + \mathcal{ O} ((\delta
t)^2).
\end{eqnarray}
Keeping only the first order terms with respect to $\delta t $ we find:
\begin{eqnarray}\label{eq:wr.24}
\rho_t ^\pm(\lambda ) &=& {\mp i  \over 2(a^\pm(\lambda ))^2 }
\biglb [\sigma _3 ,q_t(x)], \bPhi^\pm (x,\lambda ) \bigrb , \\
\label{eq:wr.25} \tau_t ^\pm(\lambda ) &=& {\pm i  \over
2(a^\pm(\lambda ))^2 } \biglb [\sigma _3 ,q_t(x)], \bPsi^\pm (x,\lambda ) \bigrb.
\end{eqnarray}

There is one more type of Wronskian relations, namely:
\begin{multline}\label{eq:wr.18'}
\left. \left( \hat{\chi }\dot{ \chi} (x,\lambda ) +ix\sigma_3 \right) \right|_{-\infty }^{\infty }\\ =
-i\int_{-\infty }^{\infty } dx\, \left(\hat{\chi }  \sigma_3\chi (x,\lambda ). -\sigma_3\right).
\end{multline}
Treating it analogously to the above cases we obtain:
\begin{multline}\label{eq:DotA}
\frac{\partial \mathcal{A}}{ \partial \lambda }\\ = -i \int_{-\infty}^{\infty} dx\;
\left( \frac{1}{2}\tr \left(  \hat{\chi }^\pm  \sigma_3\chi^\pm (x,\lambda )\sigma_3\right) -1 \right),
\end{multline}
where
\begin{equation}\label{eq:calA}\begin{split}
\mathcal{A}(\lambda) = \begin{cases} \ln a^+(\lambda)  & \lambda \in \mathbb{C}_+, \\
\frac{1}{2} \ln \left(a^+(\lambda) /a^-(\lambda)\right) & \lambda \in \mathbb{R},\\
- \ln a^-(\lambda)  & \lambda \in \mathbb{C}_-. \end{cases}
\end{split}\end{equation}
As we shall see below, this Wronskian relation is useful in analyzing the conservation laws of the
NLEE.

\section{Generalized Fourier Transforms}

Following  AKNS \cite{AKNS,Kaup} here we shall prove first that their ideas of interpreting
the ISM as a GFT hold true for the generic Zakharov-Shabat system (\ref{eq:AKNS0}) in which
$q_+$ and $q_-$ are unrelated complex-valued functions.

In what follows we assume that the potential $q(x,t)$ of the Lax operator satisfies the conditions:

\begin{description}
  \item[C1]  The potential $q(x,t)$ is a $2\times 2$ off-diagonal matrix (see (\ref{eq:AKNS0})) whose
  matrix elements are complex-valued Schwartz-type functions of $x$ and $t$;

  \item[C2] The potential $q(x,t)$ is such that the corresponding transition coefficients $a^+(\lambda)$ and
  $a^-(\lambda)$ have finite number of zeroes in their regions of analyticity located at $\lambda_k^+ \in \mathbb{C}_+$
  and  $\lambda_k^- \in \mathbb{C}_-$, $k=1,\dots, N$.
\end{description}

More specifically, from condition {\bf C1} it follows \cite{AKNS} that the matrix elements $a^\pm(\lambda)$ and
$b^\pm(t,\lambda)$ are Schwartz-type functions on the real axis of the complex $\lambda$-plane $\mathbb{C}$.

From condition {\bf C2} it follows that  $a^\pm(\lambda)$ in the vicinity of their zeroes behave like:
\begin{equation}\label{eq:apm0}\begin{aligned}
a^+(\lambda) &= (\lambda -\lambda_j^+) \dot{a}^+_j + \frac{1}{2} (\lambda -\lambda_j^+)^2 \ddot{a}^+_j +
\mathcal{O}((\lambda -\lambda_j^+)^3),  \\
a^-(\lambda) &= (\lambda -\lambda_j^-) \dot{a}^-_j + \frac{1}{2} (\lambda -\lambda_j^-)^2 \ddot{a}^-_j +
\mathcal{O}((\lambda -\lambda_j^-)^3),
\end{aligned}\end{equation}
where $ j= 1. \dots, N_+$ for the first of the above equations and $ j= 1. \dots, N_-$ for the second.

\begin{remark}\label{rem:C2}
The condition {\bf C2} requires in addition that $N_+=N_-$ which ensures that
$\lim_{\lambda\to\infty}a^\pm(\lambda) =1$, see \cite{Kaup,GeYaV}.
\end{remark}

Effectively, condition {\bf C2} means that the discrete spectrum of $L$ consists of a finite number of simple discrete eigenvalues.
The FAS $\chi^\pm(x,\lambda)$ become degenerate at the points $\lambda_k^\pm$; more specifically
\begin{equation}\label{eq:Fikpm}\begin{aligned}
\phi_k^+ (x) & = b_k^+ \psi_k^+(x), &\quad \phi_k^- (x) & = -b_k^- \psi_k^-(x),
\end{aligned}\end{equation}
where $\phi_k^\pm (x) =\phi^\pm (x,\lambda_k^\pm)$, $ \psi_k^\pm(x) =  \psi^+(x,\lambda_k^\pm)$ and $b_k^\pm$ are some
(time-dependent) constants.

\begin{remark}\label{rem:1}
In the special case, when the potential $q(x,t)$ is on a finite support one can prove that the Jost solutions, as well as the
scattering matrix $T(\lambda,t)$ are meromorphic functions of $\lambda$. Then one can extend the functions $b^+(\lambda)$
in $\mathbb{C}_+$ (resp. $b^-(\lambda)$ in $\mathbb{C}_-$) with the result:
\begin{equation}\label{eq:bkpm}\begin{aligned}
b_k^\pm (\lambda) =  b_k^\pm  + \mathcal{O}\left((\lambda -\lambda_k^\pm)\right) ,
\end{aligned}\end{equation}
i.e., the constants $b_k^\pm=b^\pm(\lambda_k^\pm)$ can be understood as the values of  $b^\pm(\lambda)$ at the points $\lambda_k^\pm$.
\end{remark}

\subsection{Completeness of the ``squared'' solutions}\label{sec:l5-2}

In this Subsection we prove that the `squared solutions' satisfy a completeness relation, thus generalizing the results of
\cite{Vg-EKh1,Vg-EKh2,Kaup}.

\begin{theorem}\label{thm:1}
Let the potential $q(x)$ of $L$ satisfy the conditions {\bf C1} and  {\bf C2}.
Then the squared solutions $\bPsi^\pm(x,\lambda )$ and $\bPhi^-(x,\lambda )$ satisfy the following
completeness relation
\begin{multline}\label{eq:CR}
\delta (x-y) \Pi_0  =
 -{1 \over \pi } \int_{-\infty }^{\infty } d\lambda \, \left( Z^+(x,y,\lambda) - Z^-(x,y,\lambda)  \right) \\
+2i \sum_{k=1}^{N_+}  X_k^+(x,y) +  2i \sum_{k=1}^{N_-}X_k^-(x,y),
\end{multline}
where
\begin{equation}\label{eq:3.19b}\begin{split}
 \Pi_0 &= \sigma _+ \otimes \sigma _- - \sigma _- \otimes \sigma _+ ,\\
Z^\pm (x,y,\lambda) &=  {\bPsi^\pm(x,\lambda ) \otimes \bPhi^\pm (y,\lambda ) \over
(a^\pm (\lambda ))^2 },
\end{split}\end{equation}
 and $ X_k^\pm(x,y) $ are defined by:
\begin{multline}\label{eq:Xpm}
X_k^\pm (x,y)   ={1  \over (\dot{a}_k^\pm)^2 } \left( \bPsi_k^\pm (x) \otimes \dot{\bPhi}^\pm_k (y)  \right. \\
\left. + \dot{\bPsi}_k^\pm (x) \otimes \bPhi^\pm_k (y) -  {\ddot{a}_k^\pm  \over \dot{a}^\pm_k }
\bPsi_k^\pm (x) \otimes \bPhi ^\pm_k (y) \right).
\end{multline}
\end{theorem}

\begin{proof}[The idea of the proof] Following \cite{Vg-EKh1,Vg-EKh2,GeYaV} we
consider the piece-wise analytic function:
 \begin{equation}\label{eq:G}
    G(x,y,\lambda) = \left\{ \begin{array}{ll}
G^+(x,y,\lambda), \quad & \mbox{for}\; \lambda \in \bbbc_+, \\
G^-(x,y,\lambda), \quad & \mbox{for}\; \lambda \in \bbbc_-, \\
    \end{array} \right.
\end{equation}
and $G(x,y,\lambda)=1/2 (G^+(x,y,\lambda)+G^-(x,y,\lambda))$ for $\lambda \in \bbbr$, where
\begin{multline}\label{eq:3.1}
G^\pm(x,y,\lambda ) = Z^\pm(x,y,\lambda ) \theta (x-y) - G_2^\pm(x,y,\lambda ) \theta (y-x),\\
G_2^\pm(x,y,\lambda ) = {1  \over (a^\pm(\lambda ))^2 } \left( \bPhi^\pm (x,\lambda ) \otimes \bPsi^\pm (y,\lambda) \right. \\
 \left. + {1  \over  2} \bTheta^\pm (x,\lambda ) \otimes \bTheta^\pm (y,\lambda ) \right).
\end{multline}

\begin{figure}[tb]
\centerline{\includegraphics*{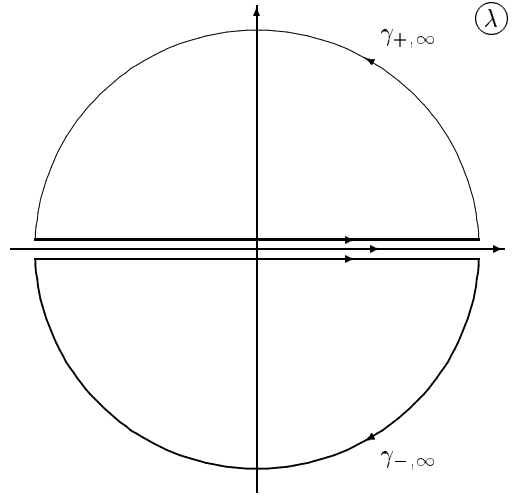}}
\caption{The contours $\gamma _\pm =\mathbb{R}\cup\gamma_{\pm\infty }$.}
\label{fig:1}

\end{figure}

Next we consider the integral
\begin{multline}\label{eq:3.3}
\mathcal{ J}_G (x,y) \\ = {1  \over  2\pi i} \left( \oint_{\gamma  _+}^{}
d\lambda \, G^+(x,y,\lambda ) - \oint_{\gamma  _-}^{} d\lambda \, G^-(x,y,\lambda ) \right) \\
= \sum_{k=1}^{N} \left( \Res_{\lambda =\lambda _k^+}
G^+(x,y,\lambda ) + \Res_{\lambda =\lambda _k^-} G^-(x,y,\lambda ) \right),
\end{multline}
where the contours $C_\pm = \mathbb{R}\cup \gamma_{\pm,\infty} $ are shown in Figure~\ref{fig:1}.

The integration along the contours gives rise to two terms: i) the integral along the real axis in eq.
(\ref{eq:CR}) which is the contribution from the continuous spectrum of $L$ to the completeness relation (CR);
ii) the integrals along the infinite semicircles can be evaluated explicitly giving rise to the term $\delta(x-y)\Pi_0$ in
eq. (\ref{eq:CR}). Evaluating the integral $\mathcal{ J}_G (x,y)$ by Cauchy residue theorem gives rise to the
contribution from the discrete spectrum of $L$ to the CR.
Indeed, the poles of $G^\pm $ coincide with $\lambda _k^\pm $.
Since $a^\pm (\lambda )$ have first order zeroes at $\lambda _k^\pm$,
then $G^\pm $ would have second order poles at these points.
\end{proof}

\subsection{The symplectic basis}\label{ssec:sb}

In this Subsection we derive an alternative form of the completeness relation (\ref{eq:CR}) and introduce
the so-called symplectic basis of `squared solutions', see \cite{Vg-EKh1,Vg-EKh2,GeYaV}. To this end, we apply to both sides of
(\ref{eq:CR}) the operator $\openone - {\bf p}$, where ${\bf p}$ acts on any tensor product $X\otimes Y$ as
follows: ${\bf p}(X\otimes Y) = Y\otimes X$. It is easy to check that $(\openone - {\bf p})\Pi_0 = 2\Pi_0$.

The same operation applied to the right hand side of (\ref{eq:CR}) can be simplified and cast into the form
\begin{multline}\label{eq:3.21a}
\delta (x-y) \Pi_0 \\ = \int_{-\infty }^{\infty } d\lambda \, \left( \mathcal{P }(x,\lambda ) \otimes \mathcal{Q }(y,\lambda ) -
\mathcal{Q }(x,\lambda ) \otimes \mathcal{P}(y,\lambda ) \right)  \\
 +  \sum_{k=1}^{N} \left( Z_k^+(x,y) +  Z_k^-(x,y) \right), \\
 Z_k^\pm (x,y) = \left( \mathcal{P }_k^\pm (x) \otimes \mathcal{Q }^\pm_k (y) -
 \mathcal{Q }_k^\pm (x) \otimes \mathcal{P }^\pm_k (y) \right) ,
\end{multline}
where we have introduced the so--called symplectic basis, which will be extensively used in
the analysis of the Hamiltonian structures  of the NLEE. The elements of this basis are
linear combinations of the ``squared'' solutions as follows:
\begin{equation}\label{eq:3.20P}\begin{split}
\mathcal{P}(x,\lambda ) &= {1  \over \pi } \left( \tau^+(\lambda ) \bPhi ^+(x,\lambda ) -  \tau^-(\lambda ) \bPhi^-(x,\lambda ) \right) \\
&= - {1  \over \pi } \left( \rho ^+(\lambda ) \bPsi^+(x,\lambda ) - \rho ^-(\lambda ) \bPsi^-(x,\lambda ) \right) ,\\
 \mathcal{P}_k^\pm (x) &= 2i C_k^\pm \bPsi_k^\pm (x) = - 2i M_k^\pm  \bPhi_k^\pm (x),
\end{split}\end{equation}
\begin{equation}\label{eq:3.20Q}\begin{split}
\mathcal{Q}(x,\lambda ) &= { \tau^+(\lambda ) \bPhi^+(x,\lambda ) +\rho ^+(\lambda ) \bPsi^+(x,\lambda ) \over 2 b^+(\lambda ) b^-(\lambda ) }  \\
 &= {\rho ^-(\lambda ) \bPsi^-(x,\lambda ) + \tau^-(\lambda ) \bPhi^-(x,\lambda ) \over 2 b^+(\lambda ) b^-(\lambda )}  ,\\
 \mathcal{Q}_k^\pm (x) &= {1  \over 2 } \left(C_k^\pm \dot{\bPsi }_k^\pm (x) + M_k^\pm \dot{\bPhi }_k^\pm (x)\right) ,
 \end{split}\end{equation}
where we must recall that $\rho^\pm(\lambda)$ and $\tau^\pm(\lambda)$ were introduced in eqs. (\ref{eq:wr.12})  and
\begin{equation}\label{eq:MC-k}
C_k^\pm = {b_k^\pm \over \dot{a}^\pm_k }, \qquad M_k^\pm = {1\over b_k^\pm \dot{a}^\pm_k }.
\end{equation}

\subsection{Expansions over the `squared' solutions}\label{sec:l5-3}

Using the completeness relations one can expand any generic
element $X(x) $ of the phase space $ \mathcal{ M} $ over each of
the three complete sets of `squared solutions'. In this section
we explain how this can be done.  We remind that $X(x) $ is a
generic element of $\mathcal{ M} $ if it is an off-diagonal
matrix--valued function, which falls off fast enough for
$|x|\to\infty  $. Obviously  $X(x)$ can be written down in terms
of its matrix elements $X_\pm(x) $ as:
\begin{eqnarray}\label{eq:4.4a}
X(x) = X_+(x) \sigma _+ + X_-(x) \sigma _-.
\end{eqnarray}
From (\ref{eq:3.19b}) we get:
\begin{multline}\label{eq:4.4b}
{1  \over 2 } \tr_2 \left(\Pi_0 [\sigma _3, X(x)] \otimes \openone \right)  \\
= -{1  \over 2 } \tr_1 \left(\openone  \otimes [\sigma _3, X(x)] \right) \Pi_0 = X(x) ,
\end{multline}
where $\tr_1 $ (and $\tr_2 $) mean that we are taking the trace of
the elements in the first (or the second) position of the tensor
product.

Now, we multiply (\ref{eq:CR}) on the right by ${1 \over 2
}[\sigma _3, X(x)] \otimes \openone  $, take $\tr_1 $ and integrate
over $dx $. This leads to the expansion of $X(x) $ over the system $\bPhi^\pm $:
\begin{multline}\label{eq:4.5a}
X(x) =  {1  \over \pi } \int_{-\infty }^{\infty } d\lambda \,
\left( \psi _X^+(\lambda ) \bPhi^+(x,\lambda ) - \psi _X^-(\lambda ) \bPhi^-(x,\lambda ) \right) \\
+ 2i \sum_{k=1}^{N} \left( \psi _{X,k}^\pm
\dot{\bPhi}_k^\pm + \dot{\psi}_{X,k}^\pm \bPhi_k^\pm \right),
\end{multline}
where
\begin{equation}\label{eq:4.5b}\begin{split}
\psi _X^\pm (\lambda ) &= {\biglb \bPsi^\pm (x,\lambda ), X(x) \bigrb \over (a^\pm(\lambda ))^2 }, \\
\psi_{X,k}^\pm &= {\biglb \bPsi^\pm_k (x) , X(x) \bigrb \over (\dot{a}^\pm_k)^2 }, \\
\dot{\psi}_{X,k}^\pm &= { 1 \over (\dot{a}^\pm_k)^2 } \Biglb \dot{\bPsi }^\pm_k (x)  - {
\ddot{a}^\pm_k \over \dot{a}^\pm_k } \bPsi ^\pm_k (x) , X(x) \Bigrb  .
\end{split}\end{equation}

Analogously, we can multiply (\ref{eq:CR}) on the left by ${1
\over 2 }\openone  \otimes [\sigma _3, X(x)] $, take $\tr_2 $ and
integrate over $dx $. This leads to the expansion of $X(x) $ over the system $\bPsi^\pm $:
\begin{multline}\label{eq:4.6a}
X(x) \\  =   - {1  \over \pi } \int_{-\infty }^{\infty } d\lambda \, \left( \phi _X^+(\lambda ) \bPsi ^+(x,\lambda ) -
\phi _X^-(\lambda ) \bPsi^-(x,\lambda ) \right) \\
+ 2i \sideset{}{^\pm}\sum_{k=1}^{N} \left( \phi _{X,k}^\pm
\dot{\bPsi }_k^\pm +\dot{\phi}_{X,k}^\pm \bPsi_k^\pm \right),
\end{multline}
\begin{equation}\label{eq:4.6b}\begin{split}
\phi _X^\pm(\lambda ) &= {\biglb \bPhi^\pm (x,\lambda ), X(x) \bigrb \over (a^\pm(\lambda ))^2 }, \\
\phi_{X,k}^\pm &= {\biglb \bPhi^\pm_k (x) , X(x) \bigrb  \over (\dot{a}^\pm_k)^2 }, \\
\dot{\phi}_{X,k}^\pm &= { 1 \over (\dot{a}^\pm_k)^2 } \Biglb \dot{\bPhi }^\pm_k (x)  - {
\ddot{a}^\pm_k \over \dot{a}^\pm_k } \bPhi ^\pm_k (x) , X(x) \Bigrb  .
\end{split}\end{equation}
The same procedure, applied to the completeness relation (\ref{eq:3.21a})
for the symplectic basis  leads to:
\begin{multline}\label{eq:4.7a}
X(x) =   \int_{-\infty }^{\infty } d\lambda \, \left( \kappa _X(\lambda ) \mathcal{P }(x,\lambda ) -
\eta _X(\lambda ) \mathcal{Q }(x,\lambda ) \right) \\
+  \sideset{}{^\pm}\sum_{k=1}^{N} \left(\kappa _{X,k}^\pm\mathcal{P}_k^\pm - \eta_{X,k}^\pm \mathcal{Q }_k^\pm \right),
\end{multline}
\begin{equation}\label{eq:4.7b}\begin{aligned}
\kappa _X(\lambda ) &= \biglb \mathcal{Q } (x,\lambda ), X(x) \bigrb , &\quad   \kappa _{X,k}^\pm &= \biglb \mathcal{Q }^\pm_k (x) ,X(x) \bigrb ,\\
\eta _X(\lambda ) &= \biglb \mathcal{P } (x,\lambda ), X(x) \bigrb , &\quad \eta_{X,k}^\pm &= \biglb \mathcal{P}^\pm_k (x) , X(x) \bigrb ,
\end{aligned}\end{equation}
where $\sideset{}{^\pm}\sum_{k=1}^{N} W^\pm = \sum_{k=1}^{N}(W^+ + W^-)$.

The completeness relations derived above allow us to establish a
one-to-one correspondence between the element $X(x) \in \mathcal{
M} $ and its expansion coefficients. Indeed, from (\ref{eq:CR})
and (\ref{eq:3.21a}) we derived the expansions (\ref{eq:4.5a}),
(\ref{eq:4.6a}) and (\ref{eq:4.7a}) with the inversion formulae
(\ref{eq:4.5b}), (\ref{eq:4.6b}) and (\ref{eq:4.7b}) respectively.
Using them we prove the following:
\begin{Prop}\label{pro:V.1}
The function $X(x)\equiv 0 $ if and only if one of the following sets of relations holds:
\begin{align} \label{eq:pr1}
\psi _X^+(\lambda ) &= \psi _X^-(\lambda ) \equiv  0,  &\;
 \psi _{X,k}^\pm &=  \dot{\psi }_{X,k}^{\pm} = 0, \\
\phi _X^+(\lambda ) &= \phi _X^-(\lambda ) \equiv  0, &\;
\label{eq:pr2} \phi _{X,k}^\pm &=  \dot{\phi }_{X,k}^{\pm} = 0, \\
 \kappa _X(\lambda ) &= \eta _X(\lambda ) \equiv  0, &\;
\kappa _{X,k}^\pm &=  \eta _{X,k}^{\pm} = 0, \label{eq:pr3}
\end{align}
where  $\lambda \in \mathbb{R}$ and $ k=1,\dots, N$.
\end{Prop}
\begin{proof}

Let us show that from $X(x)\equiv 0 $ there follows
(\ref{eq:pr1}). To this end we insert $X(x) \equiv 0 $ into the
right hand sides of the inversion formulae (\ref{eq:4.5b}) and
immediately get (\ref{eq:pr1}). The fact that from (\ref{eq:pr1})
there follows $X(x)\equiv 0 $ is readily obtained by inserting
(\ref{eq:pr1}) into the right hand side of (\ref{eq:4.5a}).

The equivalence of $X(x)\equiv 0 $ to (\ref{eq:pr2}) and (\ref{eq:pr3}) is proved analogously using the inversion formulae
(\ref{eq:4.5b}), (\ref{eq:4.6b}) and the expansions (\ref{eq:4.6a}) and (\ref{eq:4.7a}). The proposition is proved.
\end{proof}

Here we calculate the expansion coefficients for $X(x) \equiv q(x)
$. As the reader may have guessed already, their evaluation will be
based on the Wronskian relations (\ref{eq:wr.21}),
(\ref{eq:wr.22}) which we derived above. From them we have:
\begin{equation}\label{eq:4.11a}\begin{aligned}
\psi _{q}^\pm (\lambda ) &= {1  \over (a^\pm(\lambda ))^2 } \biglb \bPsi^\pm (x,\lambda ) , q(x) \bigrb = i \tau^\pm (\lambda ), \\
 \psi _{q,k}^\pm  &= {1  \over (\dot{a}^\pm_k)^2} \biglb \bPsi^\pm _k(x) , q(x) \bigrb = 0, \\
 \dot{\psi} _{q,k}^\pm  &= {1  \over (\dot{a}^\pm_k)^2 }  \Biglb \dot{\bPsi}^\pm _k(x)
 - {\ddot{a}_k^\pm  \over \dot{a}^\pm_k }  \bPsi^\pm _k(x) , q(x) \Bigrb  =  i M_k^\pm .
 \end{aligned}\end{equation}

Similarly, we evaluate all skew-scalar products between the `squared solutions' and
 $q(x,t)$, $\sigma_3 \delta q(x)$ and  $\sigma_3 \frac{\partial q(x)}{ \partial t }$, see the Tables \ref{tab:1}, \ref{tab:2}.
This means that we know the expansion coefficients of $q(x,t)$, $\sigma_3 \delta q(x)$ and  $\sigma_3 \frac{\partial q(x)}{ \partial t }$
over each of the complete sets of `squared solutions'.

\begin{table*}
  \centering
  \begin{tabular}{|c||c|c|c||c|c|c|}
   \hline \hline
$X(x,t)$  & $\phi_X^\pm  (\lambda) $ & $\phi_{X,k}^\pm $ &  $\dot{\phi}_{X,k}^\pm $ &  $\psi_{X}^\pm  (\lambda) $ &  $\psi_{X,k}^\pm $ &  $\dot{\psi}_{X,k}^\pm $ \\[5pt]
  \hline
    $q(x,t)$ & $ \pm \frac{i}{\pi} \tau^\pm (\lambda)$ & 0 & $2M_k^\pm$ & $\mp \frac{i}{\pi} \rho^\pm  (\lambda)$ & 0 & $-2C_k^\pm$ \\[5pt]
   $\sigma_3 \delta q(x)$ & $ \frac{i}{\pi}\delta \tau^\pm (\lambda)$ & $\pm 2 i\delta\lambda_k^\pm M_k^\pm $ &
              $\pm 2i \delta M_p^\pm$ & $ \frac{i}{\pi} \delta \rho^\pm  (\lambda)$ & $\pm 2i \delta \lambda_k^\pm C_k^\pm $ & $\pm 2i\delta C_k^\pm$ \\[5pt]
   $\sigma_3 \frac{\partial q}{ \partial t }$ & $ \frac{i}{\pi} \frac{\partial \tau^\pm (\lambda)}{ \partial t } $ & $\pm 2i \frac{\partial \lambda_k^\pm}{ \partial t } M_k^\pm $ &
              $\pm 2i \frac{\partial M_k^\pm}{ \partial t } $ & $ \frac{i}{\pi} \frac{\partial \rho^\pm  (\lambda)}{ \partial t} $ & $\pm 2 i \frac{\partial \lambda_k^\pm}{ \partial t}
              C_k^\pm $ & $\pm 2i \frac{\partial C_k^\pm }{ \partial t} $ \\[5pt]
    \hline
  \end{tabular}
  \caption{The expansion coefficients of $q(x,t)$, $\sigma_3 \delta q(x)$ and  $\sigma_3 \frac{\partial q(x)}{ \partial t }$
  over the sets of squared solutions $\bPsi(x,t,\lambda)$ and  $\bPhi(x,t,\lambda)$. }\label{tab:1}
\end{table*}

\begin{table}
  \centering
  \begin{tabular}{|c||c|c|c|c|}
    \hline \hline
$X(x,t)$ & $\eta_X (\lambda)$ & $\eta_{X,k}$ & $\kappa_X (\lambda)$ & $\kappa_{X,k}$  \\[5pt] \hline
    $q(x,t)$ & 0 & 0  & $i$ & $i$ \\[5pt]
    $\sigma_3 \delta q(x)$ & $i\delta \eta(\lambda) $ & $\pm \delta\lambda_k^\pm $ & $-i\delta \kappa (\lambda) $  & $\mp i\delta\ln b_k^\pm $ \\[5pt]
    $\sigma_3 \frac{\partial  q(x)}{ \partial t}$    & $i \frac{\partial \eta(\lambda)}{ \partial t} $ & $\pm  \frac{\partial \lambda_k^\pm}{ \partial t }
    $ & $-i \frac{\partial \kappa (\lambda)}{ \partial t }  $  & $\mp i \frac{\partial \ln b_k^\pm }{ \partial t }$ \\[5pt]
    \hline
  \end{tabular}
  \caption{The expansion coefficients of $q(x,t)$, $\sigma_3 \delta q(x)$ and  $\sigma_3 \frac{\partial q(x)}{ \partial t }$
  over the symplectic basis. }\label{tab:2}
\end{table}

As a result we get the following expansions:
\begin{multline}\label{eq:4.11d}
q(x)  = {i  \over \pi } \int_{-\infty }^{\infty } d\lambda \,
\left( \tau^+(\lambda ) \bPhi^+(x, \lambda ) - \tau^-(\lambda ) \bPhi^-(x, \lambda ) \right) \\ +
2\sum_{k=1}^{N} \left( M_k^+ \bPhi^+_k (x) + M_k^- \bPhi^-_k (x) \right),
\end{multline}
\begin{multline}\label{eq:4.11e}
q(x) = - {i  \over \pi } \int_{-\infty }^{\infty } d\lambda \, \left( \rho^+(\lambda ) \bPsi^+(x, \lambda ) -
\rho^-(\lambda ) \bPsi^-(x, \lambda ) \right)  \\
-2 \sum_{k=1}^{N} \left( C_k^+ \bPsi ^+_k (x) + C_k^- \bPsi ^-_k (x) \right),
\end{multline}
\begin{equation}\label{eq:4.11f}\begin{split}
q(x) = i \int_{-\infty }^{\infty } d\lambda \, \mathcal{P} (x,\lambda) + i \sum_{k=1}^{N} \left( \mathcal{P}_k^+(x) + \mathcal{P}_k^-(x) \right).
\end{split}\end{equation}

Note, that only half of the elements in the symplectic
basis contribute to the r.h.side of (\ref{eq:3.21a}). We shall see, that this makes the above basis
quite special.

The expansions  for $\sigma_3 \delta q(x) $ take the form:
\begin{multline}\label{eq:4.8d}
\sigma_3\delta q(x)  \\ = {i  \over \pi } \int_{-\infty }^{\infty } d\lambda \, \left( \delta \tau^+(\lambda ) \bPhi^+(x, \lambda ) +
\delta \tau^-(\lambda ) \bPhi^-(x, \lambda ) \right) \\
+ 2\sideset{}{^\pm}\sum_{k=1}^{N} \left(\pm  M_k^\pm \delta \lambda_k^\pm \dot{\bPhi }^\pm_k (x) \pm \delta M_k^\pm \bPhi ^\pm_k (x)  \right) ,
\end{multline}
\begin{multline}\label{eq:4.9d}
\sigma_3\delta q(x) \\ = {i  \over \pi } \int_{-\infty }^{\infty } d\lambda \, \left( \delta \rho^+(\lambda ) \bPsi^+(x, \lambda ) +
\delta \rho^-(\lambda ) \bPsi^-(x, \lambda ) \right) \\
+ 2\sideset{}{^\pm}\sum_{k=1}^{N} \left( \pm C_k^\pm \delta  \lambda_k^\pm \dot{\bPsi}^\pm_k (x)
\pm \delta C_k^\pm \bPsi^\pm_k (x) \right) ,
\end{multline}
\begin{multline}\label{eq:4.10e}
\sigma _3 \delta q(x) = i \int_{-\infty }^{\infty } d\lambda \, \left(  \delta \eta (\lambda ) \mathcal{Q} (x,\lambda )
- \delta \kappa (\lambda ) \mathcal{P} (x,\lambda ) \right) \\
+  i \sideset{}{^\pm}\sum_{k=1}^{N} \left( \delta \eta_k^\pm \mathcal{Q}_k^\pm (x) -
\delta \kappa_k^\pm \mathcal{P}_k^\pm (x)  \right),
\end{multline}
where
\begin{equation}\label{eq:4.10f}\begin{aligned}
\eta (\lambda ) &= - {1  \over \pi } \ln \left(  a ^+(\lambda) a^-(\lambda ) \right), &\;  \eta_k^\pm &= \mp 2i \lambda _k^\pm ,\\
 \kappa (\lambda ) &= {1 \over 2} \ln {b^+(\lambda) \over b^-(\lambda )} , &\;
  \kappa _k^\pm& = \pm \ln b_k^\pm .
\end{aligned}\end{equation}

The expansions over the `squared' solutions of the time
derivative $\sigma _3 q_t $  are obtained by considering
a special type of variation $\delta q(x,t) $, namely:
\begin{eqnarray}\label{eq:7.1}
\sigma _3 \delta q(x,t) = \sigma _3 q_t \delta t + \mathcal{ O}
((\delta t)^2) .
\end{eqnarray}
Keeping only the terms of order $\delta t $, from \Ref{ch:GFT}{eq:4.8d}, \Ref{ch:GFT}{eq:4.9d} and
\Ref{ch:GFT}{eq:4.10e} we find:
\begin{multline}\label{eq:7.2a}
\sigma_3 \frac{\partial q}{ \partial t }\\  = {i  \over \pi } \int_{-\infty }^{\infty }
d\lambda \, \left( \rho _t^+(t,\lambda ) {\bf \Psi }^+(x,t,\lambda ) + \rho _t^-(t,\lambda ) {\bf \Psi }^-(x,t, \lambda ) \right)  \\
+ 2\sideset{}{^\pm}\sum_{k=1}^{N} \left( \pm C_k^\pm (t) \lambda_{k,t}^\pm \dot{{\bf \Psi }}^\pm_k (x,t)
\pm C_{k,t}^\pm {\bf \Psi }^\pm_k (x,t) \right) ,
\end{multline}
\begin{multline}\label{eq:7.2b}
 \sigma_3 \frac{\partial q}{ \partial t }\\ = {i  \over \pi } \int_{-\infty }^{\infty } d\lambda \, \left( \tau_t^+(t,\lambda )
{\bf \Phi }^+(x,t, \lambda ) + \tau_t^-(t,\lambda ) {\bf \Phi }^-(x,t, \lambda ) \right) \\
+ 2\sideset{}{^\pm}\sum_{k=1}^{N} \left(\pm  M_k^\pm(t) \lambda_{k,t}^\pm \dot{{\bf \Phi }}^\pm _k (x,t)
\pm M_{k,t}^\pm {\bf \Phi }^\pm_k (x,t)  \right) ,
\end{multline}
and
\begin{multline}\label{eq:dqt}
\sigma _3  \frac{\partial q}{ \partial t }\\ = i \int_{-\infty }^{\infty } d\lambda \, \left( \eta_t (\lambda ) \mathcal{Q} (x,t,\lambda )
- \kappa_t (\lambda ) \mathcal{P} (x,t,\lambda )  \right) \\
+  i\sideset{}{^\pm} \sum_{k=1}^{N} \left( \eta_{k,t}^\pm \mathcal{Q}_k^\pm (x,t) -\kappa_{k,t}^\pm \mathcal{P}_k^\pm(x,t)  \right).
\end{multline}

\subsection{The recursion operators}

The recursion operators $\Lambda_\pm$ were introduced by AKNS \cite{AKNS} as the operators that resolve the recursion
relations that follow the compatibility condition $[L,M]=0$ of the Lax operators. Here we will use an
alternative definition which is equivalent. Namely, we will introduce $\Lambda_\pm$ as the operators for which the `squared solutions'
${\bPsi }^\pm (x,\lambda)$ and ${\bPhi }^- (x,\lambda)$ are eigenfunctions.

We start their derivation by noting that the `full' squared solutions (see eq. (\ref{eq:wr.6})) $\mathcal{E}_a^\pm (x,\lambda)
= a^\pm (\lambda)\chi^\pm \sigma_a \hat{\chi}^\pm (x,\lambda)$, where $a$ - the index of the Pauli matrices - takes values $\pm$ and $3$, satisfy the equation:
\begin{equation}\label{eq:eqE}\begin{split}
i \frac{\partial \mathcal{E}_a^\pm}{ \partial x } + \left[ q(x,t) - \lambda\sigma_3, \mathcal{E}_a^\pm(x,\lambda)\right] =0.
\end{split}\end{equation}
In view of the explicit form of the Wronskian relations we have to split the `squared solutions' into diagonal and off-diagonal parts
$ \mathcal{E}_a^{\pm} (x,\lambda) = \mathcal{E}_a^{\rm d,\pm} (x,\lambda) + \mathcal{E}_a^{\rm f,\pm} (x,\lambda) $:
\begin{equation}\label{eq:E-spl}\begin{split}
\mathcal{E}_\pm^\pm (x,\lambda) &= e^\pm_\pm (x,\lambda) \sigma_3 + \bPhi^\pm (x,\lambda), \\
\mathcal{E}_\mp^\pm (x,\lambda) &= e^\mp_\pm (x,\lambda) \sigma_3 + \bPsi^\pm (x,\lambda), \\
\mathcal{E}_3^\pm (x,\lambda) &= e^3_\pm (x,\lambda) \sigma_3 + \bTheta^\pm (x,\lambda),
\end{split}\end{equation}
where $e^\pm _a(x,\lambda) = \frac{1}{2} \tr ( \mathcal{E}_a^\pm (x,\lambda), \sigma_3)$. Inserting these splittings into eq. (\ref{eq:eqE})
we find:
\begin{equation}\label{eq:eqE2}\begin{aligned}
i \frac{\partial \mathcal{E}_a^{\rm d,\pm}}{ \partial x  } +  \left[ q(x),  \mathcal{E}_a^{\rm f, \pm} (x,\lambda) \right] &=0, \\
i \frac{\partial  \mathcal{E}_a^{\rm f,\pm} }{ \partial x } + \left[ q(x),  \mathcal{E}_a^{\rm d,\pm} (x,\lambda) \right] &=
\lambda \left[ \sigma_3,  \mathcal{E}_a^{\rm f,\pm} (x,\lambda) \right] .
\end{aligned}\end{equation}
Note that the first of the above equations does not involve $\lambda$ explicitly, so we can integrate it directly with the result:
\begin{multline}\label{eq:Ed}
\mathcal{E}_a^{\rm d,\pm}(x,\lambda) = i \sigma_3 \int_{\varepsilon \infty}^{x} dy\; \tr  \left( \sigma_3 q(y),
\mathcal{E}_a^{\rm f, \pm} (y,\lambda) \right) \\ + \lim_{x\to \varepsilon x} \mathcal{E}_a^{\rm d,\pm}(x,\lambda) ,
\end{multline}
where $\varepsilon=\pm $.
The second of the equations (\ref{eq:eqE2}) can be converted into an eigenvalue-type problem by acting on both sides by
$\frac{1}{4}\ad_{\sigma_3}$. Indeed, noting, that  $\frac{1}{4}[\sigma_3 , [\sigma_3, Z]]\equiv Z$ for any off-diagonal matrix $Z$ we find:
\begin{equation}\label{eq:Ef}\begin{split}
 \frac{i}{4} \left[ \sigma_3, \frac{\partial \mathcal{E}_a^{\rm f,\pm}}{ \partial x } \right] - q(x) e_a^{\pm} =
 \lambda \mathcal{E}_a^{\rm f,\pm} .
\end{split}\end{equation}
It remains to insert eq. (\ref{eq:Ed}) into (\ref{eq:Ef}) to obtain
\begin{equation}\label{eq:LaE}\begin{split}
(\Lambda_\varepsilon -\lambda) \mathcal{E}_a^{\rm f,\pm} = \frac{i}{2}q(x) \lim_{x\to \varepsilon \infty} \tr( \sigma_3 \mathcal{E}_a^{\rm d,\pm}(x,\lambda)),
\end{split}\end{equation}
where the recursion operator $\Lambda_\varepsilon$ is the integro-differential operator acting on any off-diagonal matrix-valued function $X=X^{\rm f}$ as:
\begin{equation}\label{eq:Laeps}\begin{split}
\Lambda_\varepsilon X \equiv  \frac{i}{4} \left[ \sigma_3, \frac{\partial X}{\partial x} \right] - i q(x)
\int_{\varepsilon\infty}^{x} \tr \left( \sigma_3 q(y) X(y) \right).
\end{split}\end{equation}
Obviously, in the cases when the limits in the right hand side of eq. (\ref{eq:LaE}) vanish, we obtain an eigenvalue equation for the
recursion operator. Skipping the details we formulate the results:
\begin{equation}\label{eq:Lala}\begin{aligned}
( \Lambda_+ - \lambda) \bPsi^\pm (x,\lambda) &=0, &\quad ( \Lambda_- - \lambda) \bPhi^\pm (x,\lambda) &=0, \\
( \Lambda_+ - \lambda_k^\pm) \bPsi_k^\pm (x) &=0, &\quad ( \Lambda_- - \lambda_k^\pm) \bPhi^\pm_k (x) &=0, \\
\end{aligned}\end{equation}
\begin{equation}\label{eq:Lala1}\begin{aligned}
( \Lambda_+ - \lambda_k^\pm) \dot{\bPsi}_k^\pm (x) &= \bPsi_k^\pm (x) , \\ ( \Lambda_- - \lambda_k^\pm) \dot{\bPhi}^\pm_k (x) &=\bPhi_k^\pm (x) ,\\
( \Lambda_\pm - \lambda) \bTheta^\pm (x,\lambda) &=i q(x) a^\pm(\lambda), \\ ( \Lambda_\pm - \lambda_k^\pm) \bTheta^\pm_k (x) &=0,
\end{aligned}\end{equation}
where we remind that $\dot{Z}_k^\pm= \frac{\partial Z}{ \partial \lambda }$ is evaluated at the corresponding point $\lambda_k^\pm$.

Note that the elements of the symplectic basis are eigenfunctions of $\Lambda = \frac{1}{2}(\Lambda_+ + \Lambda_-)$, namely:
\begin{equation}\label{eq:LaLa2}\begin{aligned}
(\Lambda - \lambda) \mathcal{P}(x,\lambda) =0, &\qquad (\Lambda - \lambda) \mathcal{Q}(x,\lambda) =0, \\
(\Lambda - \lambda_k^\pm ) \mathcal{P}_k^\pm (x) =0, &\qquad (\Lambda - \lambda_k^\pm) \mathcal{Q}_k^\pm (x) =0.
\end{aligned}\end{equation}

We will also need the `conjugation' properties of the recursion operators with respect to the skew-scalar product:
\begin{equation}\label{eq:adLa}\begin{split}
\biglb X, \Lambda_+ Y \bigrb = \biglb \Lambda_- X,  Y \bigrb, \qquad
\biglb X, \Lambda Y \bigrb = \biglb \Lambda X,  Y \bigrb,
\end{split}\end{equation}
which are easily verified with integration by parts.

\subsection{Minimal sets of scattering data}

A generic potential $q(x,t) \in \mathcal{M}^\mathbb{C}$ is determined by two complex-valued functions
$q_+(x,t)$ and $q_-(x,t)$. Solving the direct scattering problem for $L$ we find that $q(x,t)$ uniquely determines
the scattering matrix $T(\lambda,t)$ which has four matrix elements $a^\pm(\lambda)$ and  $b^\pm(\lambda,t)$ related by the
condition:
\begin{equation}\label{eq:detT}\begin{split}
\det T(\lambda,t) \equiv a^+a^-(\lambda) + b^+b^-(\lambda,t) =1.
\end{split}\end{equation}

Next we have to take into account the analyticity properties of $a^\pm(\lambda)$. We do this by considering the integral
\begin{equation}\label{eq:Ja}\begin{split}
\mathcal{J}_a(\lambda) = \frac{1}{2\pi i} \left( \int_{\gamma_+} \frac{d\mu \; \ln \tilde{a}^+(\mu)}{\mu -\lambda} -
\int_{\gamma_-} \frac{d\mu \; \ln \tilde{a}^-(\mu)}{\mu -\lambda} \right),
\end{split}\end{equation}
where
\begin{equation}\label{eq:ati}\begin{split}
\tilde{a}^+ (\lambda) &= a^+(\lambda) \prod_{k=1}^{N} \frac{\lambda -\lambda_k^-}{\lambda -\lambda_k^+}, \\
\tilde{a}^- (\lambda) &= a^-(\lambda) \prod_{k=1}^{N} \frac{\lambda -\lambda_k^+}{\lambda -\lambda_k^-},
\end{split}\end{equation}
and the contours $\gamma_\pm$ are shown in Figure \ref{fig:1}. Note that condition {\bf C2} ensures that $\tilde{a}^+(\lambda)$
(resp. $\tilde{a}^-(\lambda)$) is an analytic function for $\lambda \in \mathbb{C}_+$ (resp. for $\lambda \in \mathbb{C}_-$)
that has no zeroes in $\mathbb{C}_+$ (resp. in $\mathbb{C}_-$). Therefore, the function $\ln \tilde{a}^+(\lambda)$ (resp.
$\ln \tilde{a}^-(\lambda)$) is analytic for $\lambda \in \mathbb{C}_+$ (resp. for $\lambda \in \mathbb{C}_-$).
Let us now assume that $\lambda \in \mathbb{C}_+$ and apply the Cauchy residue theorem. Thus we get:
\begin{equation}\label{eq:Ja1}\begin{split}
\mathcal{J}_a(\lambda) = \ln \tilde{a}^+(\lambda).
\end{split}\end{equation}
Evaluating $\mathcal{J}_a(\lambda)$ by integrating along the contours splits into two terms. The first consists of integrating along
the real axis $\lambda\in \mathbb{R}$. The second term contains the integrals along the infinite semi-circles $\gamma_{\pm,\infty}$.
However the second term vanishes because both $a^\pm(\lambda)$ and  $\tilde{ a}^\pm(\lambda)$ tend to 1 for $\lambda \to \infty$,
see eq. (\ref{eq:apmlim}). Equating both results for $\mathcal{J}_a$ we obtain:
\begin{multline}\label{eq:Ja2}
\ln a^+(\lambda) = \frac{1}{2\pi i} \int_{-\infty}^{\infty} \frac{d\mu}{\mu -\lambda} \ln \left( a^+(\mu) a^-(\mu) \right) \\ +
\sum_{k=1}^{N} \ln  \frac{\lambda -\lambda_k^+}{\lambda -\lambda_k^-}.
\end{multline}

Similarly we can consider $\lambda\in \mathbb{R}$ and $\lambda\in \mathbb{C}_-$. Then by Cauchy theorem we find that $\mathcal{J}_a(\lambda)
=\mathcal{A}(\lambda)$, where $\mathcal{A}(\lambda)$ is introduced in eq. (\ref{eq:calA}). Taking into account that from $\det T(\lambda,t)=1$ there
follows $1/(a^+a^-(\lambda)) = 1+ \rho^+\rho^-(\lambda,t)$; and we obtain:
\begin{multline}\label{eq:calA2}
\mathcal{A}(\lambda) = -\frac{1}{2\pi i} \int_{-\infty}^{\infty} \frac{d\mu}{\mu -\lambda} \ln \left( 1 +\rho^+(\lambda) \rho^-(\lambda) \right)
\\ + \sum_{k=1}^{N} \ln  \frac{\lambda -\lambda_k^+}{\lambda -\lambda_k^-}.
 \end{multline}

Now we are ready to construct the minimal sets of scattering data for $L$ which determine uniquely both the potential $q(x,t)$ and
the scattering matrix.

\begin{lemma}\label{lem:1}
Let the potential $q(x,t) \in \mathcal{M}^\mathbb{C}$ and is such that the conditions {\bf C1} and {\bf C2} are
satisfied and the corresponding scattering matrix  $T(\lambda,t)$ is defined by eqs.  (\ref{eq:AKNS1}) and (\ref{eq:AKNS2}).
Then each of the sets
\begin{equation}\label{eq:MinSet}\begin{aligned}
 \mathcal{ T}_1 &\equiv \left\{ \rho ^+(\lambda,t ), \; \rho ^-(\lambda,t ),\quad \lambda \in\bbbr, \quad \lambda _k^\pm, \; C_k^\pm(t), \right\}, \\
 \mathcal{ T}_2 &\equiv \left\{ \tau ^+(\lambda ,t), \; \tau ^-(\lambda,t ),\quad \lambda \in\bbbr, \quad \lambda _k^\pm, \; M_k^\pm(t),  \right\}, \\
\mathcal{ T}_0 &\equiv \left\{ \eta(\lambda ), \; \kappa (\lambda,t ),\quad \lambda \in\bbbr, \quad \eta_k^\pm, \; \kappa_k^\pm(t),\;  \right\},
\end{aligned}\end{equation}
where $k=1,\ldots , N$ and
\begin{equation}\label{eq:ro-tau}\begin{aligned}
\rho^\pm (\lambda,t) &= \frac{b^\pm (\lambda,t)}{a^\pm (\lambda)}, &\quad C_k^\pm (t) &= \frac{b_k^\pm (t)}{\dot{a}^\pm_k},\\
\tau^\pm (\lambda,t) &= \frac{b^\mp(\lambda,t)}{a^\pm(\lambda)},  &\quad M_k^\pm (t) &= \frac{1}{b_k^\pm (t) \dot{a}^\pm_k},\\
\kappa(\lambda,t) &= \frac{1}{2} \ln \frac{b^+ (\lambda,t)}{b^- (\lambda,t)} ,  &\quad \kappa_k^\pm (t) &= \pm \ln b_k^\pm (t) ,\\
\eta(\lambda) &= \frac{1}{\pi} \ln \left( 1 +\rho^+ \rho^-(\lambda) \right)   , &\; \eta_k^\pm &=  \mp 2i \lambda_k^\pm  , &\quad
\end{aligned}\end{equation}
determines uniquely both the potential and the scattering matrix $T(\lambda,t)$ of $L$.
\end{lemma}

\begin{proof}
1. First we prove that from each $\mathcal{T}_k$ one can recover the scattering matrix $T(\lambda,t)$. Indeed, $\mathcal{T}_1$
determines uniquely the right hand side of eq. (\ref{eq:calA2}), i.e. one can recover both $a^+(\lambda)$ and $a^-(\lambda)$  as
piece-wise analytic functions of $\lambda$. Then we can easily find $b^+(\lambda)$ and $b^-(\lambda)$  since
$b^\pm(\lambda,t) = \rho^\pm(\lambda,t) a^\pm(\lambda)$. Similar arguments work out for $\mathcal{T}_2$ and  $\mathcal{T}_0$.
For  $\mathcal{T}_2$ it is enough to check that $\rho^+\rho^-(\lambda,t) = \tau^+\tau^-(\lambda,t) $ and that
$b^\pm(\lambda,t) = \tau^\mp(\lambda,t) a^\mp(\lambda)$. For  $\mathcal{T}_0$  we need the relations:
\begin{equation}\label{eq:bpm}\begin{split}
b^\pm(\lambda,t) = e^{\pm \kappa(\lambda,t)} \sqrt{1 - e^{-\pi \eta(\lambda)} }.
\end{split}\end{equation}

2. The fact that each of the sets   $\mathcal{T}_k$ (\ref{eq:MinSet}) determines uniquely the potential $q(x,t)$ is an immediate
consequence of the expansions (\ref{eq:4.11d})--(\ref{eq:4.11e}) and Proposition \ref{pro:V.1}. Indeed, the elements of   $\mathcal{T}_1$
and   $\mathcal{T}_2$ are in fact the expansion coefficients of $q(x,t)$.

\end{proof}

\section{The Fundamental Properties of the Generic NLEE}

In the next Section we shall see how this set of variables is
related to the action--angle variables of the corresponding NLEE.

The expansion (\ref{eq:4.10e}) allows us to introduce one more
minimal set of scattering data:
\begin{eqnarray}\label{eq:ms-T}
\mathcal{ T} \equiv \left\{ \eta(\lambda ), \kappa (\lambda ),
\quad \lambda \in \bbbr, \quad \eta_k^\pm, \quad \kappa _k^\pm, \right\},
\end{eqnarray}
$ k=1,\dots, N$, which, like $\mathcal{ T}_1 $ and $\mathcal{ T}_2 $ in
(\ref{eq:MinSet}), allows  to recover uniquely both the
scattering matrix $T(\lambda ) $ and the corresponding potential.
Indeed, to determine $T(\lambda ) $ from (\ref{eq:ms-T}) we make
use of the dispersion relations (\ref{eq:4.11f}), (\ref{eq:calA2}), which allow us to find
$a^\pm(\lambda ) $ in their whole domains of analyticity, knowing
$\eta(\lambda ) $ and $\lambda _k^\pm $. Then, knowing
$a^\pm(\lambda ) $ and $b^+(\lambda )/b^-(\lambda ) =\exp (2\kappa
(\lambda )) $ it is easy to determine $b^\pm(\lambda ) $ as
functions on the real $\lambda $-axis. The coefficients $C_k^\pm =
b_k^\pm /\dot{a}_k^\pm $ and $M_k^\pm = 1/(b_k^\pm \dot{a}_k^\pm)
$ are obtained through $b_k^\pm = \exp(\pm\kappa _k^\pm) $
(\ref{eq:ro-tau}). In order to determine $\dot{a}_k^+ = \left. da^+/d\lambda
\right|_{\lambda=\lambda_k^+} $  we make use of eq. (\ref{eq:apm0}) and the
dispersion relation (\ref{eq:Ja2}). Skipping the details we get:
\begin{multline}\label{eq:adotp}
\dot{a}_k^+ = \frac{1}{\lambda_k^+ -\lambda_k^-} \prod_{j\neq k}^{} \frac{\lambda_k^+ -\lambda_j^+}{\lambda_k^+ -\lambda_j^-} \\
\times \exp \left( \frac{1}{2\pi i} \int_{-\infty}^{\infty} \frac{ d\mu}{\mu - \lambda_k^+ } \ln \left( a^+a^-(\mu) \right) \right).
\end{multline}
Similarly, for $\dot{a}_k^- = \left. da^-/d\lambda \right|_{\lambda=\lambda_k^-} $ we get
\begin{multline}\label{eq:adotm}
\dot{a}_k^- = \frac{1}{\lambda_k^- -\lambda_k^+} \prod_{j\neq k}^{} \frac{\lambda_k^- -\lambda_j^-}{\lambda_k^- -\lambda_j^+} \\
\times \exp \left( -\frac{1}{2\pi i} \int_{-\infty}^{\infty} \frac{ d\mu}{\mu - \lambda_k^- } \ln \left( a^+a^-(\mu) \right) \right).
\end{multline}

The results of the previous Section allow us to solve the NLEE related to the Lax operator
$L$ (\ref{eq:AKNS0}) just like one solves PDE's with constant coefficients. The following theorem
demonstrates that.

\subsection{Description of the class of NLEEs}\label{sec:l6-2}

\begin{theorem}\label{thm:2}
Let the potential $q(x,t)$ satisfy conditions {\bf C1} and {\bf C2} and let the function $f(\lambda ) $ be meromorphic for
$\lambda\in \bbbc $ and has no singularities on the spectrum of $L$. Then the NLEEs:
\begin{eqnarray}
\label{eq:7.11c} i\sigma _3 q_t + 2 f(\Lambda ) q(x,t) = 0, \\
\label{eq:7.11a}i\sigma _3 q_t + 2 f(\Lambda _+) q(x,t) = 0, \\
\label{eq:7.11b} i\sigma _3 q_t + 2 f(\Lambda _-) q(x,t) = 0,
\end{eqnarray}
are pairwise equivalent to the following linear evolution
equations for the scattering data:
\begin{equation}\label{eq:7.12c}\begin{aligned}
i \frac{\partial \eta}{ \partial t } &= 0, &\quad  i  \frac{\partial \kappa }{ \partial t } &= 2 f(\lambda ) , \\
 i\frac{\partial  \eta^\pm_{k}}{ \partial t } &= 0, &\quad i \frac{\partial \kappa^\pm_{k}}{ \partial t } &= 2 f_k^\pm ,
\end{aligned}\end{equation}
\begin{equation}\label{eq:7.12a}\begin{aligned}
i \frac{\partial \rho ^\pm}{ \partial t }  \mp 2 f(\lambda ) \rho ^\pm (\lambda ,t) &= 0, &\quad
\frac{\partial \lambda _{k}^{\pm}}{ \partial t } C_k^\pm (t) &= 0, \\
i\frac{\partial C_{k}^\pm}{ \partial t } \mp 2 f_k^\pm C_k^\pm (t) &= 0,
\end{aligned}\end{equation}
\begin{equation}\label{eq:7.12b}\begin{aligned}
i\frac{\partial \tau ^\pm}{ \partial t }  \pm 2 f(\lambda ) \tau ^\pm (\lambda ,t) &= 0, &\quad
 \frac{\partial \lambda _{k}^{\pm}}{ \partial t }  M_k^\pm (t) &= 0, \\
i \frac{\partial M_{k}^\pm}{ \partial t } \pm 2 f_k^\pm M_k^\pm (t) &= 0,
\end{aligned}\end{equation}
where $f_k^\pm = f(\lambda_k^\pm)$.
\end{theorem}

\begin{proof}
Inserting the expansions  over the symplectic
basis (see equations (\ref{eq:dqt}) and (\ref{eq:4.11f})) into the left hand side of (\ref{eq:7.11c}) we get:
\begin{multline}\label{eq:7.13c}
 i \int_{-\infty }^{\infty } d\lambda \, \left\{ i \eta_t \mathcal{Q}
(x,t,\lambda )  - \left( i \kappa _t - 2 f(\lambda )  \right) \mathcal{P} (x,t,\lambda ) \right\} + \\
 + i \sum_{k=1}^{N} \left\{ i \eta^+_{k,t} \mathcal{Q}_k^+ (x,t) - \left( i \kappa ^+_{k,t} - 2
f_k^+ \right) \mathcal{P}_k^+ (x,t) \right.  \\
 \left.+ i \eta^-_{k,t} \mathcal{Q}_k^- (x,t) - \left( i \kappa
^-_{k,t} - 2 f_k^- \right) \mathcal{P}_k^- (x,t) \right\} =0 .
\end{multline}
Using again proposition~\ref{pro:V.1} we establish the equivalence of (\ref{eq:7.11c}) and (\ref{eq:7.12b}).

To complete the proof of the theorem it is necessary to invoke
proposition~\ref{pro:V.1} and equation (\ref{eq:ro-tau}), from which
it follows that the l.h. sides of the NLEEs (\ref{eq:7.11a}) -- (\ref{eq:7.11c}) coincide.  The theorem is proved.

\end{proof}

\begin{remark}\label{rem:l6-1}
Note the special  role of the symplectic basis and the related scattering
data \Ref{ch:GFT}{eq:ms-T}. From (\ref{eq:7.12c}) we see that half of these data namely,
$\eta(\lambda ) $, $\eta_k^\pm $ are time-independent, while
 the other half $\kappa (\lambda ) $, $\kappa _k^\pm $, depend linearly
on time.  This implies that $\mathcal{ T} $ in fact provides us
with global action--angle variables for the NLEEs (\ref{eq:7.11c}) -- (\ref{eq:7.11b}).
We shall return to this question in Section~\ref{sec:l6-3} below.

\end{remark}

\subsection{Examples  of NLEEs with polynomial dispersion laws}\label{sec:l6-3}

Here we list several examples of physically important NLEEs which
fall into the above scheme. Theorem \ref{thm:2} shows, that
each NLEEs is specified by the corresponding function $f(\lambda )
$. In physics this function is known as the dispersion law of the
NLEEs; clearly $f(\lambda ) $ fixes up uniquely both the explicit
form of the NLEEs and the evolution of the scattering
data. Below we list a few examples of NLEE with  `polynomial in
$\lambda $' dispersion laws.

In order to find the explicit form of the NLEEs we shall need to
calculate $\Lambda_\pm^p q(x,t) $ for $p=1,2,3$. The calculation shows that:
\begin{align}\label{eq:7.15}
 \Lambda_\pm q(x,t) &= {i  \over 4 } [\sigma _3, q_x], \\
\label{eq:7.16}
 (\Lambda_\pm)^2 q(x,t) &= -{1 \over 4 } \left(  q_{xx} + 2 q_+ q_- q(x,t) \right), \\
\label{eq:7.17}
 (\Lambda_\pm)^3 q(x,t) &= -{i \over 16 } \left[\sigma _3,  q_{xxx} + 6 q_+
q_- q_x \right] .
\end{align}
These expressions illustrate two important facts. The first one
was actually introduced by (\ref{eq:ro-tau}); the second one, which
will be analyzed below, is that $\Lambda ^k q(x,t) $ are local in
$q(x,t) $ for positive $k $, i.e. $\Lambda ^k q(x,t) $ depend only
on $q $ and its $x $-derivatives.

The generic NLEEs will be systems of equations for the two
complex-valued functions $q_+(x,t) $ and $q_-(x,t) $, which
parametrize the potential $q(x,t) $. Next, we shall show how these
systems of NLEEs can be simplified by additional algebraic
restrictions on $q(x,t) $. Let us consider some examples.

\begin{example}[The GNLS equation.]\label{exa:nls}
This generalization of the NLS equation is obtained by choosing
$f(\lambda ) = -c_2\lambda ^2 $. Then (\ref{eq:7.11a}) and
(\ref{eq:7.16}) lead to the following system:
\begin{equation}\label{eq:7.18a}\begin{aligned}
 i \frac{\partial q_+}{ \partial t}  + \frac{c_2}{2} \frac{\partial^2 q_+}{ \partial x^2 } + 2 c_2 (q_+)^2 q_-(x,t)) &= 0, \\
- i \frac{\partial q_-}{ \partial t} +\frac{c_2}{2} \frac{\partial^2 q_-}{ \partial x^2 } + 2c_2 (q_-)^2 q_+(x,t)) &= 0 .
\end{aligned}\end{equation}
\end{example}

\begin{example}[The GmKdV equation.]\label{exa:mkdv}
The dispersion law for the generalized mKdV equation is given by
$f(\lambda ) = -4c_3 \lambda ^3 $. Then (\ref{eq:7.11a}) and  (\ref{eq:7.17}) lead to:
\begin{equation}\label{eq:7.19a}\begin{aligned}
 \frac{\partial q_+}{ \partial t} + c_3 \frac{\partial^3 q_+}{ \partial x^3 } + 6c_3 q_+ \frac{\partial q_+}{ \partial x } q_-(x,t)) &= 0, \\
 \frac{\partial q_-}{ \partial t} +  c_3 \frac{\partial^3 q_-}{ \partial x^3 } + 6c_3 q_- \frac{\partial q_-}{ \partial x } q_+(x,t)) &= 0 .
\end{aligned}\end{equation}
\end{example}

\begin{example}[Mixed GNLS--GmKdV equation.]\label{exa:nls-mkdv}
The last example here is a generalization of the NLEEs with a
dispersion law  $f(\lambda ) = -c_2\lambda ^2 - 4c_3 \lambda ^3 $,
where $c_2$ and $c_3  $ are some real constants.  The corresponding system of NLEEs is:
\begin{equation}\label{eq:7.20b}\begin{aligned}
 i \frac{\partial q_+}{ \partial t} &+ \frac{c_2}{2} \frac{\partial^2 q_+}{ \partial x^2 } + c_2 q_+^2 q_-(x,t) \\
& + i c_3 \left( \frac{\partial^3 q_+}{ \partial x^3 } + 6 q_+ \frac{\partial q_+}{ \partial x } q_-(x,t)\right)  = 0, \\
- i \frac{\partial q_-}{ \partial t} &+ \frac{c_2}{2}\frac{\partial^2 q_-}{ \partial x^2 } + c_2 q_-^2 q_+(x,t) \\
&- i c_3 \left( \frac{\partial^3 q_-}{ \partial x^3 } + 6 q_- \frac{\partial q_-}{ \partial x } q_+(x,t)\right)  = 0 .
\end{aligned}\end{equation}
\end{example}

It is well known, that each of these equations can be further simplified by imposing algebraic constraints
on $q_\pm$. For example, putting $c_2=1$ and $q_+(x,t) = q_-^*(x,t) = u(x,t) $ eq. (\ref{eq:7.18a}) goes
into the famous NLS eq. \cite{ZaSha}. Similarly, putting $q_+(x,t)=q_-(x,t)=v(x,t)$ eq. (\ref{eq:7.20b}) becomes
the modified KdV for $v(x,t)$. We will return to this point later.

\subsection{The Integrals of motion and trace identities}

Here we start with the generating functional of the integrals of motion $\mathcal{A}(\lambda)$ (\ref{eq:DotA}).
Using the contour integration method, see \cite{AKNS,Vg-EKh1,Vg-EKh2,GeYaV} we find that it can be expressed in terms
of the minimal set of scattering data as follows (see Remark \ref{rem:C2} and eq. (\ref{eq:apmlim})):
\begin{multline}\label{eq:calA'}
\mathcal{A}(\lambda) = \frac{1}{2\pi i} \int_{-\infty}^{\infty} \frac{ d\mu}{\mu -\lambda} \ln (a^+(\mu) a^-(\mu)) \\
+ \sum_{k=1}^{N} \ln \left( \frac{ \lambda - \lambda_k^+}{\lambda - \lambda_k^-} \right).
\end{multline}

As integrals of motion we will consider the expansion coefficients of $\mathcal{A}$ over the inverse powers of $\lambda$:
\begin{equation}\label{eq:A-Cp}\begin{split}
 \mathcal{A} (\lambda) = i \sum_{p=1}^{\infty} \frac{C_p}{\lambda^p}.
\end{split}\end{equation}
From (\ref{eq:calA'}) and (\ref{eq:A-Cp}) it follows that
\begin{multline}\label{eq:Cp-sd}
 C_p  = \frac{1}{2\pi}  \int_{-\infty}^{\infty}  d\mu \;\mu^{p-1} \ln (a^+(\mu) a^-(\mu)) \\ - \frac{1}{p}
 \sum_{k=1}^{N} \left( (\lambda_k^+)^p - (\lambda_k^-)^p \right)
 = -\frac{1}{2}  \int_{-\infty}^{\infty}  d\mu \;\mu^{p-1} \eta (\mu) \\ - \frac{i^p}{2^p p}
 \sum_{k=1}^{N} \left( (\eta_k^+)^p - (-\eta_k^-)^p \right) .
\end{multline}

Besides, from the Wronskian relation (\ref{eq:wr.18'}) we derive their dependence on the potential $q(x,t)$
and its $x$-derivatives. Indeed, using eqs. (\ref{eq:DotA}) and the last line of (\ref{eq:Lala1}) we obtain:
\begin{equation}\label{eq:dA1}\begin{split}
\frac{\partial \mathcal{A}}{ \partial \lambda } &=  - i \int_{-\infty}^{\infty} dx\; \left( \frac{1}{a^\pm (\lambda)}
\tr \left( \bTheta^\pm (x,\lambda) \sigma_3 \right) -1 \right) \\
& = \int_{-\infty}^{\infty} dx\; \left( \tr \left(  (\Lambda_\pm -\lambda)^{-1} q(x)\sigma_3 \right) +i \right).
\end{split}\end{equation}
It remains to compare the expansions of the right hand sides of (\ref{eq:dA1}) and (\ref{eq:calA'}) over the inverse
powers of $\lambda$. As a result we obtain the following compact expression for $C_p$ in terms of $q$ and its derivatives
\begin{equation}\label{eq:Cp-q}\begin{split}
C_p = \frac{1}{ip} \int_{-\infty}^{\infty} dx\;  \tr \left(  \Lambda_\pm^p q(x)\sigma_3 \right).
\end{split}\end{equation}

The well known trace identities \cite{AKNS,FaTa,Vg-EKh2} follow immediately by equating the right hand sides of eqs.
(\ref{eq:Cp-sd}) and (\ref{eq:Cp-q}).
Let us list the first few nontrivial conserved quantities in terms of $q(x)$
\begin{multline}\label{eq:C1}
C_1 = -\frac{1}{2} \int_{-\infty}^{\infty} dx\; q_+ q_-(x,t) \\
= -\frac{1}{2} \int_{-\infty}^{\infty}  d\mu \eta (\mu)  +  \frac{1}{2}\sum_{k=1}^{N} \left( \eta_k^+ + \eta_k^-\right),
\end{multline}
\begin{multline}\label{eq:C2}
C_2 = -\frac{i}{8} \int_{-\infty}^{\infty} dx\; \left(  q_+  \frac{\partial q_-}{ \partial x }- q_- \frac{\partial q_+}{ \partial x }\right) \\
= -\frac{1}{2 } \int_{-\infty}^{\infty}  d\mu \mu \eta(\mu)   + \frac{i}{8}\sum_{k=1}^{N} \left( (\eta_k^+)^2 - (\eta_k^-)^2\right),
\end{multline}
\begin{multline}\label{eq:C3}
C_3 = -\frac{1}{8} \int_{-\infty}^{\infty} dx\; \left(  - \frac{\partial q_+}{ \partial x }\frac{\partial q_-}{ \partial x }
+ (q_-  q_+)^2\right)\\
= -\frac{1}{2 } \int_{-\infty}^{\infty}  d\mu \mu^2 \eta (\mu)  - \frac{1}{24}\sum_{k=1}^{N} \left( (\eta_k^+)^3 + (\eta_k^-)^3\right),
\end{multline}
and
\begin{multline}\label{eq:C4}
 C_4 = \frac{i}{32} \int_{-\infty}^{\infty} dx\; \left(   \frac{\partial q_+}{ \partial x } \frac{\partial^2 q_-}{ \partial x^2 }-
  \frac{\partial^2 q_+}{ \partial x^2 } \frac{\partial q_-}{ \partial x } \right. \\
  \left.  - 3 q_+q_- \left( q_+\frac{\partial q_-}{ \partial x }-  q_-\frac{\partial q_+}{ \partial x }    \right) \right)\\
= -\frac{1}{2} \int_{-\infty}^{\infty}  d\mu \mu^3 \eta (\mu)   - \frac{i}{64}\sum_{k=1}^{N} \left( (\eta_k^+)^4 - (\eta_k^-)^4\right) .
\end{multline}
One can prove that the densities of all $C_p$ are local in $q(x)$, i.e. depend only on
$q(x)$ and its $x$-derivatives.

\subsection{The generic NLEEs as  Complex Hamiltonian System}\label{ssec:l6-8}

Here we start  with the notions of a complexified phase space and Hamiltonian.
We can introduce  canonical Poisson brackets between any two functionals on $\mathcal{M}^\bbbc$ by:
\begin{multline}\label{eq:19.15}
\{ F, G\}_{(0)}^\bbbc \\ = i \int_{-\infty }^{\infty } dx \left( {\delta F  \over \delta q_-(x) } {\delta G  \over \delta q_+(x) }
- {\delta F\over \delta q_+(x) } {\delta G  \over \delta q_-(x) }\right) ,
\end{multline}
where both $F $ and $G $ are complex-valued functionals on
$\mathcal{M}^\bbbc $ depending analytically on $q_\pm(x) $.
Then the corresponding canonical symplectic form can be written as:
\begin{eqnarray}\label{eq:v7-0me}
\Omega _{(0)}^\bbbc = {1 \over i} \int_{-\infty }^{\infty } dx\, \delta q_-(x) \wedge \delta q_+(x).
\end{eqnarray}

Next we need to specify the Hamiltonian $H^\mathbb{C}$ as a functional over $\mathcal{M}^\bbbc$. Skipping the
details (see \cite{GKMV1,GKMV2,GeYaV}) we note that  $H^\mathbb{C}$ must depend analytically on
the fields $q_\pm (x,t)$ and their $x$-derivatives.

The generic Hamiltonian equations of motion generated by $H^\bbbc
$ and the Poisson brackets (\ref{eq:19.15}) are the following:
\begin{equation}\label{eq:19.16}\begin{split}
{\partial q_+  \over \partial t } &= \{ H^\bbbc, q_+(x,t)\}_{(0)}^\bbbc =i{\delta H^\bbbc \over \delta q_-(x) }, \\
 -{\partial q_-  \over \partial t } &= -\{H^\bbbc, q_-(x)\}_{(0)}^\bbbc = i{\delta H^\bbbc \over \delta q_+(x) } .
\end{split}\end{equation}
They are equivalent to a standard Hamiltonian system (with twice more
real dynamical variables and degrees of freedom), provided
$H^\bbbc $ is analytic with respect to $q_+ $ and $q_- $. The
analyticity of $H^\bbbc $ means that its real and imaginary parts
$H_0^\bbbc $ and $H_1^\bbbc $ satisfy the analog of Cauchy-Riemann equations:
\begin{equation}\label{eq:19.17}
{\delta H_0^\bbbc  \over \delta q^0_\pm(x) } = {\delta H_1^\bbbc
\over \delta q^1_\pm(x) }, \qquad {\delta H_1^\bbbc  \over \delta
q^0_\pm(x) } = -{\delta H_0^\bbbc  \over \delta q^1_\pm(x) },
\end{equation}
where $q^{a }_\pm $, $a =0,1 $ are the real and imaginary parts of $q_\pm (x) $:
\begin{equation}\label{eq:19.18}
q_\pm (x,t) = q^0_\pm (x,t) + i q^1_\pm (x,t).
\end{equation}
Then eqs.  (\ref{eq:19.16}) go into:
\begin{equation}\label{eq:19.19}\begin{aligned}
{\partial q^0_+ \over \partial t} &= {\delta H_{(0)}^\bbbc  \over \delta q^1_-(x) },  &\quad
 {\partial q^1_+ \over \partial t} &= {\delta H_{(0)}^\bbbc  \over \delta q^0_-(x) }, \\
{\partial q^0_-  \over \partial t} &= -{\delta H_{(0)}^\bbbc  \over \delta q^1_+(x) }, &\quad
 {\partial q^1_-  \over \partial t} &= -{\delta H_{(0)}^\bbbc  \over \delta q^0_+(x) },
\end{aligned}\end{equation}
which can be viewed as the equation of motion of an infinite-dimensional Hamiltonian system  with real-valued Hamiltonian
$H_{(0)}^\bbbc $.  The elements of the phase space $\mathcal{M}^\bbbc $ can be viewed also as the 4-tuples of real
functions $\{ q^0_+, q^1_+, q^0_-, q^1_- \} $ vanishing fast enough for $x\to\pm\infty $. We designate the space of such $4
$-tuples by $\mathcal{M}_\bbbr $.
The  Poisson brackets on $\mathcal{M}_\bbbr$  correspond to the canonical symplectic form  on $\mathcal{M}_\bbbr  $
defined by the real part of $\Omega _{(0)}^\bbbc $:
\begin{equation}\label{eq:Re-ome0}
\Re \Omega _{(0)}^\bbbc = \int_{-\infty }^{\infty } dx\, \left(
\delta q^0_+ \wedge \delta q^1_- + \delta q^1_+ \wedge \delta
q^0_- \right).
\end{equation}

In what follows we shall use the formal Hamiltonian formulation of
the generic NLEEs with complex-valued Hamiltonian $H^\bbbc $ and complex valued
dynamical fields $q^\pm(x) $ as in eq. (\ref{eq:19.16}). Those who
prefer the standard Hamiltonian formulations  using real-valued Hamiltonians
 and dynamical fields can always rewrite the generic
NLEEs in their equivalent form (\ref{eq:19.19}).

\begin{remark}\label{rem:1a}
Similarly we can use an equivalent Hamiltonian formulation with $H=H_{(1)}^\mathbb{C}$ and
\begin{equation}\label{eq:Im-ome0}
\mbox{\rm Im\;} \Omega _{(0)}^\bbbc = \int_{-\infty }^{\infty } dx\, \left( \delta q^1_- \wedge \delta q^1_+ - \delta q^0_- \wedge \delta
q^0_+ \right).
\end{equation}
In what follows we will use the previous formulation with the Poisson brackets:
\begin{subequations}\label{eq:cPB}
\begin{equation}\label{eq:cPBa}\begin{aligned}
\{ q^a_+(x,t),  q^b_+(y,t)\}_{(0)} &=0, &\quad \{ q^a_-(x,t),  q^b_-(y,t)\}_{(0)} &=0, \\
\{ q^0_-(x,t),  q^0_+(y,t)\}_{(0)} &=0, &\quad \{ q^1_-(x,t),  q^1_-(y,t)\}_{(0)} &=0,  \\
\end{aligned}\end{equation}
where $ a,b =1,2,$ and
\begin{equation}\label{eq:cPBb}\begin{aligned}
\{ q^0_-(x,t),  q^1_+(y,t)\}_{(0)} &=\frac{1}{2}\delta(x-y) , \\
 \{ q^1_-(x,t),  q^0_+(y,t)\}_{(0)} &=\frac{1}{2}\delta(x-y).
\end{aligned}\end{equation}
\end{subequations}
\end{remark}

In order to adapt the Hamiltonian formulation better to the gauge covariant approach to the `squared solutions'
used above, below we shall view the phase space $\mathcal{M}^\bbbc $ as the
space of $2\times 2 $ off-diagonal matrices $q(x) = \left(\begin{array}{cc} 0 & q_+\\ q_- & 0 \end{array}\right) $. The
variational derivatives (or the `gradients') of the functional $H^\bbbc $ then will be written as:
\begin{equation}\label{eq:19.22}
\nabla_q H^\bbbc \equiv {\delta H^\bbbc  \over \delta q^T(x) }
= \left( \begin{array}{cc} 0 & {\delta H^\bbbc \over \delta q_-(x)} \\
{\delta H^\bbbc \over \delta q_+(x)} & 0  \end{array} \right).
\end{equation}

It remains to recall the definition of the skew-scalar product $\biglb \cdot
\;, \; \cdot \bigrb $ and after simple calculation one is able to
write down the canonical Poisson brackets (\ref{eq:19.15}) as follows:
\begin{multline}\label{eq:19.23}
\{ F, G\}_{(0)}^\bbbc = {i \over 2} \int_{-\infty }^{\infty } \tr \left( \nabla_q F, \left[ \sigma _3, \nabla_q G \right]
\right) dx \\  =  i  \biglb \nabla_q F, \nabla _q G \bigrb .
\end{multline}

The corresponding canonical symplectic form $\Omega_0^\bbbc $  becomes:
\begin{multline}\label{eq:19.29}
\Omega _0^\bbbc = i \int_{-\infty }^{\infty } dx\, \left( \delta q_+(x) \wedge\delta q_-(x)\right) \\
= {i  \over 2 } \biglb \sigma _3\delta q \wedgecomma \sigma _3 \delta q \bigrb .
\end{multline}
By the symbol $\wedgecomma  $ above we mean that we first perform
the matrix multiplication keeping the order of the factors, and
then replace the standard multiplication by an exterior
product $\wedge $.
With all these notations we can write down (\ref{eq:19.16}) in the form:
\begin{equation}\label{eq:19.21}
i \sigma_3 {\partial q\over\partial t} +\nabla_q H^\bbbc = 0.
\end{equation}

The system (\ref{eq:7.18a})   generalizing the
NLSE can be written down as a complex Hamiltonian system (\ref{eq:19.21}) with
$H^\bbbc $ chosen to be:
\begin{eqnarray}\label{eq:8.11}
H^\bbbc = \int_{-\infty }^{\infty } dx\, \left[ q_x^- q_x^+ - (q_+
q_-)^2 \right]  = -8 C_3.
\end{eqnarray}

Quite analogously one may check that all the other examples of
NLEEs also allow complex Hamiltonian structures with the
symplectic structure introduced on $\mathcal{ M}_\bbbc $ by (\ref{eq:19.22}).

Each of the generic NLEE (\ref{eq:7.11c}) with dispersion law
$f(\lambda ) $ can be written down in the form (\ref{eq:19.21}).
It is only natural to expect that the corresponding Hamiltonian
$H$ should be expressed in terms of the integrals of motion
$ C_p $. Indeed, eq. (\ref{eq:dAlam}) can be written down in the form:
\begin{equation}\label{eq:19.24}
\nabla_q C_p = - {1  \over 2 }\Lambda ^{p-1} q(x).
\end{equation}
Then if we choose
\begin{equation}\label{eq:19.25}
H^\bbbc = \sum_{k>0}^{} 4 f_k C_{k+1},
\end{equation}
we get:
\begin{equation}\label{eq:19.26}
\nabla _q H^\bbbc =  -2 f(\Lambda ) q(x).
\end{equation}
Thus Eq.~(\ref{eq:19.21}) coincides with the NLEEs
(\ref{eq:7.11c}) with the dispersion law $f(\lambda) = \sum_{k>0} f_k\lambda^k$.

\subsection{Action-angle variables of the generic NLEE}

The most straightforward way to derive the action--angle variables of the NLEE
(\ref{eq:19.21}) is to insert into the r.h. side of (\ref{eq:19.29})
the expansion (\ref{eq:4.10e}) for $\sigma _3 \delta q(x) $. This gives:
\begin{multline}\label{eq:9.4}
\Omega ^\bbbc_{(0)}\\ = i\int_{-\infty}^{\infty} dx\; \delta q_+(x,t) \wedge \delta q_-(x,t)
\equiv \frac{i}{2}\biglb \sigma_3 \delta q \wedgecomma \sigma_3 \delta q \bigrb \\
= {-1  \over 2 } \Bigglb \sigma _3\delta q(x) \wedgecomma \left(  \int_{-\infty }^{\infty }\!\!\! d\lambda \, \left(
\delta \eta(\lambda) {\bQ}(x,\lambda )  - \delta \kappa (t,\lambda) {\bP}(x,\lambda )\right) \right.  \\
\left. + \sideset{}{^\pm} \sum_{k=1}^{N} (\delta \eta_k^\pm {\bQ}_k^\pm (x) - \delta \kappa ^\pm_k {\bP}^\pm _k(x)) \right)\Biggrb  \\
= i \int_{-\infty }^{\infty } d\lambda \, \delta \eta(\lambda) \wedge \delta \kappa (t,\lambda) + i
\sideset{}{^\pm} \sum_{k=1}^{} \delta \eta _k^\pm \wedge \delta \kappa _k^\pm ,
\end{multline}
where $\kappa(\lambda)$, $\eta(\lambda)$, $\kappa_k^\pm$ and $\eta_k^\pm$ are given in eq. (\ref{eq:ro-tau}).

In the above calculation we made use of the inversion formulae for
the symplectic basis (\ref{eq:4.7a}) and
identified the skew-symmetric scalar products of $\sigma _3\delta
q(x) $ with the elements of the symplectic basis with the
variations of $\delta \eta (\lambda) $, $\delta \kappa (t,\lambda) $, etc.

From (\ref{eq:9.4}) we see also, that the 2-form $\Omega _0^\bbbc$ expressed by the variables
$\eta(\lambda)$, $\kappa(\lambda)$, $\eta_k^\pm $ and $\kappa
_k^\pm $ has a canonical form. Recall now the trace identities
(\ref{eq:Cp-sd}). From them it follows, that the Hamiltonian
$H^\bbbc $ of the NLEE depends only on the variables $\eta(\lambda
) $, $\eta_k^\pm $:

\begin{equation}\label{eq:H0}\begin{aligned}
H^\bbbc &=  - 2 \int_{-\infty }^{\infty } d\mu \, f(\mu) \eta (\mu) + 4i \sum_{k=1}^{N} \left(F_k^+ - F_k^-  \right), \\
 f(\mu) &= \sum_{k>0}^{} f_p \mu^{p-1},
\end{aligned}\end{equation}
where
\begin{equation}\label{eq:H0a}\begin{aligned}
F_k^+ &= F(i\eta _k^+/2), \qquad F_k^- = F(-i\eta _k^-/2), \\ F(\lambda ) &= \int_{}^{\lambda } d\lambda ' f(\lambda ').
\end{aligned}\end{equation}
Obviously  $H^\mathbb{C}$ depends only on the action variables $\eta(\lambda)$ and $\eta_k^\pm$.
Let us now recall one of the results of Theorem \ref{thm:2}, stating that if $q(x,t) $ satisfies the NLEE (\ref{eq:19.21})
then the variables $\eta (\lambda) $, $\kappa (\lambda,t)$, $\eta_k^\pm $ and $\kappa _k^\pm (t)$ satisfy:
\begin{equation}\label{eq:AA-var}\begin{aligned}
{d\eta  \over dt }&=0, &\quad {d\eta_k^\pm  \over dt } &=0, \\
i{d\kappa \over dt }&= 2f(\lambda ),  &\quad i{d\kappa _k^\pm  \over dt } &=2f(\lambda_k^\pm ).
\end{aligned}\end{equation}
Thus
\begin{description}
  \item[AA1] the variables $\eta (\lambda) $, $\kappa (\lambda,t)$, $\eta_k^\pm $ and $\kappa _k^\pm (t)$
  form a canonical basis in $\mathcal{M}^\bbbc $;
  \item[AA2] Hamiltonian $H^\mathbb{C}$ depends only on `half' of them: $\eta(\lambda)$ and $\eta_k^\pm$, $k=1,\dots, N$;
  \item[AA3] The variables $\eta(\lambda)$ and $\eta_k^\pm$, $k=1,\dots, N$ are time independent, while  $\kappa (\lambda,t)$
   and $\kappa _k^\pm (t)$  depend linearly on $t$.
\end{description}
Thus we conclude that these variables are generalized action-angle variables of the generic NLEE (\ref{eq:19.21}).
Another reason to use the term generalized is in the fact, that these variables are complex-valued.

Of course, they  can be written as:
\begin{equation}\label{eq:AA-03}\begin{aligned}
\eta (\lambda) & = \eta_0(\lambda) + i \eta_1(\lambda) , &\quad \kappa(\lambda,t) &=\kappa_0(\lambda,t)
+i \kappa_1(\lambda,t),  \\
\eta_k^\pm & = \eta_k^{0,\pm}+ i\eta_k^{1,\pm} , &\quad \kappa_k^\pm (t) &=\kappa_k^{0,\pm} + i \kappa_k^{1,\pm} ,
\end{aligned}\end{equation}
where
\begin{subequations}\label{eq:AA-04}
\begin{equation}\label{eq:AA-04a}\begin{aligned}
\eta_0(\lambda,t) & = - \frac{1}{\pi} \ln | a^+a^-(\lambda)|, \\
 \eta_1(\lambda,t) & = - \frac{1}{\pi} \left( \arg ( a^+(\lambda) ) + \arg (a^-(\lambda)) \right),  \\
\kappa_0(\lambda,t) &= \frac{1}{2} \ln \left| \frac{b^+(\lambda,t)}{b^-(\lambda,t)}\right| , \\
\kappa_1(\lambda,t) &=  \frac{1}{2} \left( \arg ( b^+(\lambda,t) ) - \arg (b^-(\lambda,t)) \right),
\end{aligned}\end{equation}
\begin{equation}\label{eq:AA.04b}\begin{aligned}
 \eta_k^{0,\pm} &= \mp 2 \lambda_k^{1,\pm}, &\quad  \eta_k^{1,\pm} &= \mp 2 \lambda_k^{0,\pm}, \\
\kappa_k^{0,\pm} &= \pm \ln |b_k^\pm|, &\quad \kappa_k^{1,\pm} &= \pm \arg (b_k^\pm).
\end{aligned}\end{equation}
\end{subequations}

Note that the derivation of this result is based on the
completeness relation  of the symplectic basis. This
ensures:  (i)~the uniqueness and the invertibility of the mapping
from $\{q^\pm(x)\} $ to $\mathcal{T}$; (ii)~the nondegeneracy of
the 2-form $\Omega^\bbbc_{(0)}  $ on $\mathcal{M}^\bbbc $.

One can view $q_\pm(x) $ as local coordinates on
$\mathcal{M}^\bbbc $; any functional $F $ or $G $ on
$\mathcal{M}^\bbbc $ can be expressed in terms of $q^\pm(x) $. The
variations $\delta F $ and $\delta G $ of the functionals $F$ and
$G$ are the analogs of $1 $-forms over $\mathcal{M}^\bbbc $. They
can be expressed in terms of the `gradients' by:
\begin{equation}\label{eq:grad-1}
\delta F = \biglb \nabla_q F, \sigma _3\delta q \bigrb , \qquad
\delta G = \biglb \nabla_q G, \sigma _3\delta q\bigrb .
\end{equation}
The `gradients' $\nabla_q F $ and $\nabla_q G $ are elements of
the space $T_q\mathcal{M}^\bbbc $ tangential to $\mathcal{M}^\bbbc$.

At the same time the mapping to $\mathcal{T} $ is one-to-one,
therefore it is possible to express $F $  and $G $ in terms of the
scattering data. To this end we consider
the expansions of $\nabla_q F $ and $\nabla_q G $ over the symplectic basis:
\begin{multline}\label{eq:grad-F}
\nabla_q F = i\int_{-\infty }^{\infty} \left( \eta_{F}(\lambda) \bQ(x,\lambda ) -\kappa _{F}(\lambda )\bP(x,\lambda ) \right) d\lambda \\
+ i \sideset{}{^\pm}\sum_{k=1}^{N} \left( \eta_{F,k}^\pm \bQ_k^\pm(x) -\kappa _{F,k}^\pm \bP_k^\pm(x) \right),
\end{multline}
\begin{equation}\label{eq:grad-Fc}\begin{aligned}
\eta_{F}(\lambda )&= i \biglb \bP(x,\lambda ), \nabla_q F \bigrb, &\quad  \eta_{F,k}^\pm &= i \biglb \bP_k^\pm (x), \nabla_q F \bigrb,\\
\kappa _{F}(\lambda ) &= i \biglb \bQ(x,\lambda ), \nabla_q F \bigrb,   &\quad \kappa _{F,k}^\pm &= i \biglb \bQ_k^\pm (x), \nabla_q F \bigrb.
\end{aligned}\end{equation}
Similar expansion for $\nabla_q G $ is obtained from (\ref{eq:grad-F}) by changing $F $ to $G $. Such expansions will
hold true provided $F $ and $ G $ are restricted in such a way that the expansion coefficients $\eta_{F}(\lambda ) $ and $\kappa
_{F}(\lambda ) $ are smooth and fall off fast enough for $\lambda \to\pm\infty  $. In what follows we shall assume that the
functionals $F $ and $G $ satisfy the following:

\medskip
{\bf Condition C3.} The functionals $F $ and $G $ are restricted
by the following implicit condition: the expansion coefficients
$\eta_{F}(\lambda ) $ and $\kappa _{F}(\lambda ) $ and
$\eta_{G}(\lambda ) $ and $\kappa _{G}(\lambda ) $ are
Schwartz-type functions of $\lambda  $ for real $\lambda  $.

\medskip
Using the bi-quadratic relations satisfied by the elements of the symplectic basis \cite{GeYaV} we can express the
Poisson brackets between $F $ and $G $ in terms of their expansion coefficients as follows:
\begin{multline}\label{eq:F-G-PB}
\{F,G\}_{(0)}^\bbbc  \\ \equiv  -i \biglb \nabla_q F, \nabla_q G \bigrb
= \int_{-\infty }^{\infty }d\lambda \, \left(\eta_{F}\kappa _{G} -\kappa _{F}\eta_{G}\right)(\lambda ) \\ +
\sideset{}{^\pm}\sum_{k=1}^{N} \left(\eta_{F,k}^\pm\kappa _{G,k}^\pm -\kappa_{F,k}^\pm \eta_{G,k}^\pm\right).
\end{multline}
In particular, if we choose $F=H^\bbbc $ then from eq. (\ref{eq:H0}) we find that $\eta_{H^\bbbc}(\lambda )=0 $,
$\eta_{H^\bbbc,k}^\pm=0 $ and
\begin{equation}\label{eq:nabH0}\begin{split}
\kappa _{H^\bbbc,k}(\lambda )= -2f(\lambda ), \qquad \kappa _{H^\bbbc,k}^\pm=-2f(\lambda _k^\pm),
\end{split}\end{equation}
which gives:
\begin{multline}\label{eq:H-G-PB}
\{H,G\}_{(0)}^\bbbc = 2\int_{-\infty }^{\infty } d\lambda \, f(\lambda ) \eta_G(\lambda ) \\
 + 2 \sum_{k=1}^{N} ( f(\lambda_k^+)\eta_{G,k}^+ + f(\lambda _k^-)\eta_{G,k}^-).
\end{multline}
Eq. (\ref{eq:H-G-PB}) allows us to describe implicitly the set
of functionals $G $ that are in involution with all integrals of
motion of the generic NLEE. Indeed, the right hand side of (\ref{eq:H-G-PB}) will vanish
identically for all choices of the dispersion law $f(\lambda ) $
only if the expansion coefficients of $\nabla_q G $ satisfy:
\begin{equation}\label{eq:H-G=0}
\eta_g(\lambda )=0, \qquad \lambda \in \bbbr; \qquad \eta_{G,k}^{\pm} =0, \qquad \forall k=1,\dots, N.
\end{equation}

We end this subsection by noting the special role of the subspace $\mathcal{L}^\bbbc\subset \mathcal{M}^\bbbc $ spanned by
$\bP(x,\lambda ) $ and $\bP_k^\pm(x) $, $k=1,\dots,N  $, i.e. by `half' of the elements of the symplectic basis. Obviously all Hamiltonian vector
fields  with Hamiltonians of the form (\ref{eq:19.25}) induce dynamics which is tangent to $\mathcal{L}^\bbbc $. This is a maximal subspace of
$\mathcal{M}^\bbbc $ on which the symplectic form $\Omega_{(0)}^\bbbc $ is degenerate.
Therefore $\mathcal{L}^\bbbc $ is the Lagrange submanifold of $\mathcal{M}^\bbbc $. Dynamics defined by the Hamiltonian vector fields
that are not tangential violate the complete integrability.

\section{Hamiltonian hierarchies and action-angle variables}

\subsection{The Hierarchy of Poisson brackets}\label{ssec:9.2}

The complete integrability of the generic NLEE makes them rather
special. They have an infinite number of integrals of motion $C_n$ which are in
involution and moreover satisfy the relation:
\begin{equation}\label{eq:Len2}
\nabla_q C_{n+m} = \Lambda ^m \nabla_q C_{n},
\end{equation}
which generalizes the Lenard relation (\ref{eq:19.24}). The important fact here is that the
recursion operator $\Lambda  $ is universal one and does not
depend on either $n $ or $m $. This has far reaching consequences
which we outline below.

The first one consists in the possibility to introduce a hierarchy
of Poisson brackets:
\begin{eqnarray}\label{eq:hPB}
\{ F, G \}_{(m)}^\bbbc &=& {1 \over i}\biglb \nabla_q F, \Lambda
^m \nabla_q G \bigrb .
\end{eqnarray}
Below we shall show that these Poisson brackets satisfy all the
necessary properties.

First, using the fact that $\Lambda  $ is `self-adjoint' with respect to the skew-symmetric scalar product (\ref{eq:adLa}) we
easily check that the Poisson bracket defined by (\ref{eq:hPB}) is skew-symmetric. Indeed:
\begin{equation}\label{eq:hPB-1}\begin{split}
 \{ F, G \}_{(m)}^\bbbc &= {1 \over i} \biglb \nabla_q F, \Lambda ^m \nabla_q G \bigrb = -  {1 \over i}\biglb \Lambda ^m  \nabla_q G, \nabla_q F \bigrb \\
&= -{1 \over i} \biglb \nabla_q G, \Lambda ^m \nabla_q F \bigrb =-\{ G, F\}_{(m)}^\bbbc .
\end{split}\end{equation}

We have also the Leibnitz rule:
\begin{multline}\label{eq:hPB-Leib}
\{ FG, H \}_{(m)}^\bbbc =  {1 \over i}\biglb \nabla_q (F G), \Lambda ^m \nabla_q H \bigrb \\
= {1 \over i} F \biglb  \nabla_q G, \Lambda ^m \nabla_q H \bigrb
+{1 \over i}\biglb  \nabla_q F, \Lambda ^m \nabla_q H \bigrb G \\
= F\{ G,H\}_{(m)}^\bbbc + \{ F,H\}_{(m)}^\bbbc G,
\end{multline}
since  $\nabla_q (FG)= F\nabla_q G +(\nabla_q F) G $.

Using the expansion (\ref{eq:grad-F}) of $\nabla_q F $, an analogous one for $\nabla_q G $ and the fact that the elements
$\bP(x,\lambda ) $ and $\bQ(x,\lambda ) $ are eigenfunctions of $\Lambda  $ (see eq. (\ref{eq:LaLa2})) we find:
\begin{multline}\label{eq:F-G-PBm}
\{ F,G\}_{(m)}^{\bbbc} = -i \biglb \nabla_q F, \Lambda ^m \nabla_q G \bigrb \\
= \int_{-\infty }^{\infty } d\lambda \, \lambda ^m \left(\eta_F \kappa_G - \kappa _F \eta_G \right) (\lambda )
\\ + \sideset{}{^\pm}\sum_{k=1}^{N} (\lambda _k^\pm)^m
\left(\eta_{F,k}^\pm \kappa _{G,k}^\pm  -\kappa _{F,k}^\pm \eta_{G,k}^\pm \right).
\end{multline}

The Jacobi identity  is far from trivial to check in these terms. This will be done using the corresponding
symplectic form and for that reason we postpone the proof until later.

The existence of a hierarchy of Poisson brackets entails that
there must exist also hierarchy of vector fields, symplectic forms, etc.

\subsection{The Hierarchy of symplectic forms}
Let us define:
\begin{eqnarray}\label{eq:9.7}
\Omega _{(m)}^\bbbc = {i  \over 2 } \biglb \sigma _3\delta
q(x)\wedgecomma \Lambda ^m \sigma _3 \delta q(x) \bigrb .
\end{eqnarray}
These 2-forms are not canonical. The proof of the fact that
$\delta \Omega_{(m)}^{\bbbc}=0 $ is performed by recalculating
them in terms of the `action-angle'
variables. For this we follow the same idea as in the calculation of $\Omega^\bbbc_{(0)}
$, see eq. (\ref{eq:9.4}). We insert the expansion for $\sigma _3
\delta q(x)$ over the symplectic basis and then act on this expansion by $\Lambda ^m $.  This is easy to
do because of (\ref{eq:LaLa2}) and the result is:
\begin{multline}\label{eq:9.8}
\Lambda ^m \sigma _3\delta q(x) \\ = i \int_{-\infty }^{\infty } d\lambda \, \lambda ^m \left( \delta \eta(\lambda ) {\bQ}
(x,\lambda ) - \delta \kappa (t,\lambda) {\bP}(x,\lambda ) \right) \\
+ i\sideset{}{^\pm} \sum_{k=1}^{N} \left( (\lambda _k^\pm)^m \left( \delta \eta_k^\pm {\bQ}_k^\pm (x) - \delta \kappa _k^\pm {\bP}_k^\pm (x) \right) \right) ,
\end{multline}
where $\kappa(\lambda)$, $\eta(\lambda)$, $\kappa_k^\pm$ and $\eta_k^\pm$ are given in eq. (\ref{eq:ro-tau}).

Calculating the skew-symmetric scalar products of $\sigma _3 \delta q(x)$ with the r.h. side  of (\ref{eq:9.8}), we again obtain
the variations of the $\eta $ and $\kappa  $--variables. Finally we get:
\begin{multline}\label{eq:9.9}
 \Omega^\bbbc_{(m)}  =   i \int_{-\infty }^{\infty } d\lambda \, \lambda ^m \delta \kappa (t,\lambda) \wedge \delta \eta(\lambda ) \\
  + i \sum_{k=1}^{N} \left( (\lambda _k^+)^m \delta \kappa _k^+ \wedge
\delta \eta _k^+ + (\lambda _k^-)^m \delta \kappa _k^- \wedge \delta \eta _k^- \right) \\
= i \int_{-\infty }^{\infty } d\lambda \, \lambda ^m \delta \kappa (t,\lambda) \wedge \delta \eta(\lambda )\\
  + c_{0,m} \sum_{k=1}^{N} \left(  \delta \kappa _k^+ \wedge \delta (\eta _k^+)^{m+1} + (-1)^m
   \delta \kappa _k^- \wedge \delta (\eta _k^-)^{m+1}  \right),
\end{multline}
where $c_{0,m}=i^{m+1}2^{-m} (m+1)^{-1}$.

\begin{remark}\label{rem:7.4}
The right hand sides of (\ref{eq:9.9}) are well defined for all
$m\geq 0 $ for potentials $q(x) $ satisfying condition {\bf C1}.
This condition ensures that $\kappa (t,\lambda ) $  and
$\eta(\lambda ) $ are Schwartz-type functions of $\lambda  $.
\end{remark}

\begin{remark}\label{rem:7.5}
For negative values of $m $ the existence of the integrals in (\ref{eq:9.9}) is ensured only, provided we put additional
restrictions on $q(x) $ which would ensure that $\lim_{\lambda \to 0} \lambda ^m \delta \kappa (t,\lambda)
\wedge \delta \eta(\lambda) $ exists for all $m<0 $.
\end{remark}

Now it is easy to prove
\begin{Prop}\label{prop:7.4} Let the potential $q(x) $ satisfy condition {\bf C1}.  Then each of the symplectic forms
$\Omega^\bbbc_{(m)} $ for $m\geq 0 $ is closed, i.e.
\begin{equation}\label{eq:Ome_m-0}
\delta \Omega^\bbbc_{(m)} =0, \qquad m=0,1,2,\dots.
\end{equation}
If in addition $q(x) $ satisfies the condition in remark \ref{rem:7.5} then each of the symplectic forms
$\Omega^\bbbc_{(m)} $ is closed also for $m<0 $.

\end{Prop}

\begin{proof}
Indeed, the condition in proposition \ref{prop:7.4} is such that
the integral in the right hand side of (\ref{eq:9.9}) is well
defined so we can interchange the integration with the operation
of taking the external differential $\delta  $. Therefore we have:
\begin{multline}\label{eq:prop-7.4}
\delta \Omega^\bbbc_{(m)}  = i \int_{-\infty }^{\infty } d\lambda \, \lambda ^m \delta \left(\delta \kappa (t,\lambda) \wedge \delta \eta(\lambda )\right) \\
+ c_{0,m} \sum_{k=1}^{N}  \delta \left(  \delta \kappa _k^+ \wedge \delta (\eta _k^+)^{m+1} + (-1)^m
   \delta \kappa _k^- \wedge \delta (\eta _k^-)^{m+1}  \right) ,
\end{multline}
where we used the simple fact that $\delta (\delta g(\lambda))\equiv 0 $ for any $g(\lambda ) $.
\end{proof}

\begin{corollary}\label{cor:7-4}
Direct consequence of Proposition \ref{prop:7.4} is that the
Poisson brackets $\{\, \cdot , \,\cdot \}_{(m)}^\bbbc $ satisfy Jacobi identity.
\index{Jacobi identity}

\end{corollary}

Therefore we have shown that each generic NLEE allows a hierarchy of Hamiltonian formulations:
\begin{equation}\label{eq:19.16mb}\begin{split}
{\partial q_+  \over \partial t } &= \{ H_{(p)}^\bbbc, q_+(x)\}^\bbbc_{(-p)} , \\
 -{dq_-  \over dt } &= -\{ H_{(p)}^\bbbc, q_-(x)\}^\bbbc _{(-p)},
 \end{split}\end{equation}
where
\begin{equation}\label{eq:H_p}\begin{split}
H_{(p)}^\mathbb{C}  = \sum_{k>0}^{} 4f_k C_{k+p+1}, \qquad  \nabla_q H_{(p)}^\mathbb{C}
= -2\Lambda^p f(\Lambda) q(x,t),
\end{split}\end{equation}
see eqs. (\ref{eq:19.25}), (\ref{eq:19.26}).

One can relate the existence of the hierarchy of Hamiltonian structures  to the simple fact that the generic NLEE:
\begin{eqnarray}\label{eq:9.10b}
\Lambda ^m \left( i \sigma _3 \pd{q},{t} + 2 f(\Lambda ) q(x,t) \right) = 0,
\end{eqnarray}
are equivalent to the NLEE (\ref{eq:19.21}), (\ref{eq:19.25}). In terms of the
`action--angle' variables
\begin{equation}\label{eq:9.12}\begin{aligned}
 i \lambda ^m \eta_t &= 0, &\qquad  \lambda ^m( i \kappa _t - 2 f(\lambda)) &=0, \\
 i (\lambda_k^\pm) ^m \eta^\pm_{k} &= 0,  &\qquad  (\lambda_k^\pm)^m( i \kappa^\pm _{k,t} - 2 f(\lambda_k^\pm) )& = 0.
\end{aligned}\end{equation}

For $m=0 $ we recover the equations from section 4.

\section{Local and nonlocal involutions of the  Zakharov--Shabat system}\label{ssec:l6-3.3}

\subsection{The (local) involution $q_-(x,t)=(q_+)^*(x,t) =u(x,t)$ }\label{ssec:10.1'}

It is well known that one can impose on the generic Zakharov-Shabat system $L$ additional constraints.
These are achieved with the help of Cartan-like involutions. In this way one can derive different real
Hamiltonian forms of the generic NLEE, see \cite{GKMV1,GKMV2,GKMV3}.

We will outline this procedure for the best known involutions which allowed Zakharov and Shabat to solve
the NLS equation.

The involution $q_+(x,t)=(q_-(x,t))^* =u(x,t)$ is a consequence of  the following symmetry of $L$ and,
consequently on $U(x,t,\lambda)$:
\begin{equation}\label{eq:U*}\begin{split}
U^*(x,t,\lambda^*)  =- \sigma^{-1} U(x,t,\lambda) \sigma, \qquad \sigma = \left(\begin{array}{cc} 0 & 1 \\ -1 & 0
 \end{array}\right).
\end{split}\end{equation}
As a consequence of this symmetry we obtain constraints on the FAS and on the scattering matrix as follows:
\begin{equation}\label{eq:FasT1}\begin{aligned}
(\chi^+(x,t,\lambda^*))^* & = \sigma^{-1} \chi^-(x,t,\lambda) \sigma , \\
 T^*(\lambda^*,t)  &= \sigma^{-1} T(\lambda,t) \sigma,
\end{aligned}\end{equation}
or in components:
\begin{equation}\label{eq:apm1}\begin{aligned}
a^-(\lambda) &= (a^+(\lambda^*))^*, &\qquad b^-(\lambda,t) &= (b^+(\lambda^*,t))^*.
\end{aligned}\end{equation}
Thus if $\lambda_k^+$ is a zero of $a^+(\lambda)$ then $\lambda_k^{+,*}$ is a zero of $a^-(\lambda)$. In other words,
the discrete eigenvalues of $L$ satisfying this reduction come in complex conjugate pairs $\lambda_k^{-}=\lambda_k^{+,*}$.

More specifically, for the action-angle variables we get:
\begin{subequations}\label{eq:AA*1}
\begin{equation}\label{eq:AA*1a}\begin{aligned}
\eta(\lambda) &=\eta^*(\lambda)= - \frac{1}{\pi}\ln |a^+(\lambda)|^2 , \\
\kappa(\lambda,t) &= -\kappa^*(\lambda,t)  =i \arg b^+(\lambda) ,
\end{aligned}\end{equation}
\begin{equation}\label{eq:AA*1b}\begin{aligned}
  \eta_k^+ & =(\eta_k^-)^* =-2i \lambda_k^+, &\quad b_k^- &= b_k^{+,*}, \\
 \kappa_k^+ &= -(\kappa_k^-)^* = \ln |b_k^+| +i \arg b^+_k, &\quad \dot{a}_k^- &= (\dot{a}_k^+)^*,
\end{aligned}\end{equation}
\end{subequations}
where $ \lambda \in \mathbb{R} $ and $k=1,\dots, N$.
Thus we conclude that $\Omega_{(0)}$ becomes purely real and has the form:
\begin{multline}\label{eq:10.17}
\Omega_{(0)} =\int_{-\infty }^{\infty} d\lambda\,\delta \eta(\lambda )\wedge \delta \arg b^+(t,\lambda)  \\
+ 1\sum_{k=1}^{N} \left( \delta \eta _k^+ \wedge \delta \kappa^+_k (t) - \delta \eta _{k}^{+,*} \wedge \delta \kappa_k^{+,*}(t)\right).
\end{multline}

As regards the integrals of motion, they also become real valued:
\begin{multline}\label{eq:10.19}
C_p = - {1 \over 2 } \int_{-\infty }^{\infty } d\lambda \, \lambda ^{p-1} \eta(\lambda )\\
 + \frac{ i^{p+1}}{2^p p}  \sum_{k=1}^{N} \left(  (\eta _k^+)^p - (-\eta _k^{+,*})^p \right)
\end{multline}
and so does the Hamiltonian $H^{(0)} $
\begin{multline}\label{eq:H-red1}
H_{(0)}= 4 \sum_{p}^{} f_pC_{p+1} = -2\int_{-\infty }^{\infty } d\mu \, f(\mu ) \eta(\mu ) \\
-8\sum_{k=1}^{N} \Im (F_k^+),
\end{multline}
where $f(\lambda ) $ is the dispersion law and $F(\lambda ) $ and
$F_k^+ $ are introduced in (\ref{eq:H0a}).

Similar effects hold true also for  the family of symplectic forms of the hierarchy
$\Omega_{(m)} $:
\begin{multline}\label{eq:10.18}
\Omega_{(m)} = \int_{-\infty }^{\infty } d\lambda \, \lambda ^m \delta \eta(\lambda ) \wedge \delta \arg b^+(t,\lambda)  \\
+ c_{0,m} \sum_{k=1}^{N} \left( \delta (\eta_k^+)^{m+1} \wedge \delta \kappa_k^+ - (-1)^m
\delta (\eta_k^{+,*})^{m+1} \wedge \delta \kappa_k^{+,*} \right).
\end{multline}
The hierarchy of Hamiltonians  is provided by:
\begin{multline}\label{eq:Hh-red1}
H_{(m)} = 4 \sum_{p}^{} f_pC_{p+m+1} \\
 = -2\int_{-\infty }^{\infty } d\mu \, \mu ^m f(\mu ) \eta(\mu ) + 4i \sum_{k=1}^{N} (F_{(m),k}^{+} -F_{(m),k}^{-}) ,
\end{multline}
where
\begin{equation}\label{eq:Fkmp}
F^{(m)}(\lambda) = \int^\lambda d\lambda'\; f(\lambda')
\lambda^{\prime,m}, \quad  F_{(m),k}^{\pm} =  F_{(m)}\left( \frac{\pm i\eta_k^\pm}{2} \right). \quad
\end{equation}

Thus we see that the overall effect of the reduction is to
decrease `twice' the number of dynamical variables both on the
continuous and discrete spectrum. Now two of the three types of
action variables: $\arg
b^+(t,\lambda) $ and $\arg b_k^+ $ are real and take values in the
interval $[0,2\pi] $; the third type $\ln |b_k^+| $ is also real,
but may take arbitrary values.

The reduction imposes  also restrictions on the dispersion law of
the NLEE. In other words the reduction
(\ref{eq:U*}) admits only dispersion laws whose expansion
coefficients are real:
\begin{equation}\label{eq:f-red1a}
f(\lambda ) = \sum_{p}^{} f_p\lambda ^p, \qquad f_p=f_p^*.
\end{equation}
Some of the most important examples of NLEE obtained by this
reduction ($u=q_+(x,t) =q_-^*(x,t)$) are the NLS eq., the complex mKdV eq.  and a combination
of both NLS and complex mKdV:
\begin{equation}\label{eq:NLEE-red11}\begin{aligned}
& \mbox{NLS}  &\; & iu_t +u_{xx} + 2|u|^2u(x,t) =0, \\
& \mbox{cmKdV:}  &\;  & u_t +u_{xxx} + 6|u|^2u_x(x,t) =0 , \\
& \mbox{NLS-cmKdV:} &\;  &iu_t +u_{xx} + 2|u|^2u(x,t) \\
& &\; &\qquad  + ic_0 (u_{xxx} + 6|u^2|u_x)=0.
\end{aligned}\end{equation}
Their dispersion laws  are given by:
\begin{equation}\label{eq:disp-red1}\begin{split}
f_{\rm NLS}(\lambda ) &= -\lambda ^2, \qquad f_{\rm cmKdV}(\lambda ) = -4 \lambda ^3,  \\
f_{\rm NLS-cmKdV}(\lambda ) &= -\lambda ^2 - 4c_0\lambda ^3.
\end{split}\end{equation}

\subsection{The nonlocal involution A: $q_-(-x,t)=(q_+)^*(x,t) =u(x,t)$ }\label{ssec:10.1}

The involution stated in the title of this Subsection is one of the simplest nonlocal involution compatible
with our Lax pair. If we assume that $q_-(x,t) =   q_+(-x,t)^*$ then one can check that
$U(x,t,\lambda) =q(x,t)-\lambda \sigma_3$ satisfies:
\begin{equation}\label{eq:Ue*1}\begin{split}
U(x,t,\lambda) = \sigma_1 \left( U(-x,t,-\lambda^*)\right)^* \sigma_1 ^{-1},
\qquad \sigma_1  = \left(\begin{array}{cc} 0 &  1 \\ 1 & 0  \end{array}\right).
\end{split}\end{equation}

In other words the involution on the potential $q(x,t)$ can be extended to an involutive automorphism
of the affine Lie algebra in which $U(x,t,\lambda)$ and $V(x,t,\lambda)$ take values. Obviously this
involutive automorphism extends to all fundamental solutions of $L$. For example, the Jost solutions
must satisfy \cite{AblMus}:
\begin{equation}\label{eq:Jo*1}\begin{split}
\psi^-(x,t,\lambda) &= \sigma_1 ^{-1} \left( \phi^- (-x,t,-\lambda^*)\right)^*, \\
\psi^+(x,t,\lambda) &= \sigma_1  \left( \phi^+ (-x,t,-\lambda^*)\right)^*.
\end{split}\end{equation}
This relation can be written in compact form:
\begin{equation}\label{eq:Jo*2}\begin{split}
\psi(x,t,\lambda) = \sigma_1  \left( \phi (-x,t,-\lambda^*)\right)^* \sigma^{-1} .
\end{split}\end{equation}
Taking into account eq. (\ref{eq:AKNS2})  we derive the following constraint for the scattering matrix:
\begin{equation}\label{eq:T*0}\begin{split}
T(\lambda,t) = \sigma_1  \left(  \widehat{T}(-\lambda^*,t) \right)^* \sigma_1 ^{-1};
\end{split}\end{equation}
similarly from eq. (\ref{eq:Fas0}) for the FAS we find:
\begin{equation}\label{eq:chiR1}\begin{split}
\chi^+(x,t,\lambda)  &= \sigma_1  \left( \chi^+ (-x,t,-\lambda^*)\right)^* \sigma_1 ^{-1}  , \\
\chi^-(x,t,\lambda)  &= \sigma_1  \left( \chi^- (-x,t,-\lambda^*)\right)^* \sigma_1 ^{-1}  .
\end{split}\end{equation}

The involution imposes a condition also on the dispersion law of the NLEE; in our case this is:
\begin{equation}\label{eq:f*}\begin{split}
f(\lambda) &= (f(-\lambda^*))^*, \quad \mbox{i.e.} \quad f(\lambda) = \sum_{k>0}^{}  f_{2k} \lambda^{2k}  ,
\end{split}\end{equation}
which means that $f(\lambda)$ must be an even polynomial of $\lambda$ with real coefficients. In particular, this nonlocal involution
A) is compatible with the NLS equation, because its dispersion law is $f_{\rm NLS}(\lambda)=-2\lambda^2$ but
can not be applied to the mKdV eq. since its dispersion law is $f_{\rm mKdV}(\lambda)=-4\lambda^3$.
In what follows we will need eq. (\ref{eq:T*0}) in components:
\begin{equation}\label{eq:ab*1}\begin{split}
a^+(\lambda) &= (a^+(-\lambda^*))^*, \qquad a^-(\lambda) = (a^-(-\lambda^*))^*, \\
b^-(t,\lambda) &=   (b^+(t,-\lambda^*))^*.
\end{split}\end{equation}

The scattering data on the discrete spectrum is also subject to constraints. For example,
the constraints on $a^\pm (\lambda)$ (\ref{eq:ab*1}) mean that if $\lambda_k^+$ is a zero of
$a^+(\lambda)$ (resp. if $\lambda_k^-$ is a zero of $a^-(\lambda)$) then $-(\lambda_k^+)^*$
is also a zero of $a^+(\lambda)$ (resp. $-(\lambda_k^-)^*$ is also a zero of $a^-(\lambda)$).
In other words, the discrete eigenvalues of $L$ must be symmetrically situated around the
imaginary axis in the complex $\lambda$-plane.

\begin{corollary}\label{cor:4}
The above statement means that $L$ must have two different types of discrete eigenvalues:
a) purely imaginary ones
\begin{equation}\label{eq:ins}\begin{split}
\lambda_s^+ = - (\lambda_s^+)^*= in_s^+, \qquad \lambda_s^- =- (\lambda_s^-)^*= -in_s^-   ;
 \end{split}\end{equation}
 with $n_s^\pm$ real and positive and $ s=1,\dots N_1$;
b) pairs of complex  conjugated ones
\begin{equation}\label{eq:irpm}\begin{aligned}
\lambda_r^+ &=p_r^+ + is_r^+, &\quad  \lambda_{r+N_2}^+& =- \lambda_r^{+,*}=-p_r^+ + is_r^+ , \\
\lambda_r^- &=p_r^- - is_r^-, &\quad \lambda_{r+N_2}^- &=- \lambda_r^{-,*}=-p_r^- - is_r^- ,
\end{aligned}\end{equation}
where $ r=N_1+1,\dots ,N_1+N_2$,  $p_r^\pm  >0$ and $s_r^\pm >0$.
In what follows without restrictions we assume that $n_s^\pm $ for $s=1,\dots , N_1$ and $p_r$, $s_r$ for
$r=N_1+1,\dots, N_2$ are all real and positive.
Then the number of  discrete eigenvalues of $L$ will be $N= 2N_1+4N_2$. In particular,
the two different types of eigenvalues mean that the nonlocal NLS, like the sine-Gordon equation
 must have two different types  of soliton solutions.

\end{corollary}

\begin{remark}\label{rem:5}
The simplest one-soliton solution for the nonlocal NLS was calculated in \cite{AblMus} and shown to
be a singular function which does not vanish for $x\to\pm \infty$, see also \cite{Val}. The soliton
in \cite{AblMus} however was derived assuming that $a^+(\lambda)$ has only one simple zero
which is purely imaginary, thus violating condition {\bf C2}. Our hypothesis is that one may be able
to construct regular soliton solutions if we take configurations of the zeroes  $\lambda_k^\pm$
satisfying condition {\bf C2}. In particular this means, that the simplest regular soliton solution will
be a two-soliton one.
\end{remark}

The discrete spectrum of $L$ satisfying the non-local reduction consists of all
the zeroes of $a^\pm (\lambda)$ which can be of two sorts,  see Corollary  \ref{cor:4} above.

\begin{lemma}\label{lem:2}
The functions  $a^\pm(\lambda)$ satisfy the dispersion relations
\begin{multline}\label{eq:CalA2}
 \mathcal{A}(\lambda) = \frac{1}{2\pi i} \int_{-\infty}^{\infty} \frac{d\mu}{\mu - \lambda} \ln \left( a^+a^-(\mu) \right) \\
 + \sum_{s=1}^{N_1} \ln \frac{ \lambda -in_s^+}{\lambda +in_s^-}  +\sum_{r=N_1+1}^{N_1+N_2}
 \ln \frac{(\lambda - \lambda_r^+)(\lambda +\lambda_r^{+,*}) }{(\lambda - \lambda_r^-)(\lambda +\lambda_r^{-,*})}.
\end{multline}
where
\begin{equation}\label{eq:calA3}\begin{aligned}
 \mathcal{A}(\lambda) = \begin{cases} \ln a^+(\lambda) & \mbox{for} \quad \lambda \in \mathbb{C}_+ \\
 \frac{1}{2} \ln (a^+(\lambda)/a^-(\lambda)) & \mbox{for} \quad \lambda \in \mathbb{R} \\
 - \ln a^-(\lambda) & \mbox{for} \quad \lambda \in \mathbb{C}_- .
 \end{cases}
\end{aligned}\end{equation}
\end{lemma}

\begin{proof}

The proof goes in analogy with one of the generic cases in Subsection IV.3.
 Only now we need a different definition of $\tilde{a}^\pm (\lambda)$ as follows:
\begin{equation}\label{eq:ta2pm}\begin{split}
\tilde{a}^+ (\lambda) = a^+(\lambda) \prod_{s=1}^{N_1} \frac{\lambda +in_s^-}{\lambda -in_s^+}
 \prod_{r=1}^{N_2} \frac{(\lambda -\lambda_r^-)(\lambda +\lambda_r^{-,*}) }{(\lambda -\lambda_r^+)(\lambda +\lambda_r^{+,*})}, \\
\tilde{a}^- (\lambda) = a^-(\lambda) \prod_{s=1}^{N_1} \frac{\lambda -in_s^+}{\lambda +in_s^-}
 \prod_{r=1}^{N_2} \frac{(\lambda -\lambda_r^+)(\lambda +\lambda_r^{+,*}) }{(\lambda -\lambda_r^-)(\lambda +\lambda_r^{-,*})},
\end{split}\end{equation}
where $\tilde{a}^\pm$ are analytic functions for $\lambda\in \mathbb{C}_\pm$ and have no zeroes in their regions of analyticity.
We can again evaluate the integral $\mathcal{J}_a(\lambda)$ as in (\ref{eq:Ja}) and derive dispersion relations for the
analytic functions $a^\pm(\lambda)$. Due to the different reduction we get a new result given by (\ref{eq:CalA2}) above.
\end{proof}

Taking into account that
\begin{equation}\label{eq:apam}\begin{split}
a^+a^-(\mu) = (1+\rho^+\rho^-(\mu))^{-1} =(1+\tau^+\tau^-(\mu))^{-1}
\end{split}\end{equation}
 one easily concludes that the analytic functions
$a^+a^-(\mu)$ can be recovered easily from each of the minimal sets of scattering data (\ref{eq:MinSet}); again we have to
take into account the different structure of the set of discrete eigenvalues.

We will need also the behavior of $a^\pm (\lambda)$ in the neighborhood of the discrete eigenvalues:
\begin{equation}\label{eq:apm3}\begin{split}
a^\pm (\lambda)  \mathop{\simeq}\limits_{\lambda \simeq \lambda_s^\pm } (\lambda \mp in_s^\pm) \dot{a}^\pm_s
&+ \frac{1}{2}  (\lambda \mp in_s^\pm)^2 \ddot{a}^\pm_s \\
&\qquad + \mathcal{O}\left( (\lambda \mp in_s^\pm)^3\right), \\
a^\pm  (\lambda)  \mathop{\simeq}\limits_{\lambda \simeq \lambda_r^\pm } (\lambda - \lambda_r^\pm ) \dot{a}^\pm _r
&+ \frac{1}{2}  (\lambda - \lambda_r^\pm )^2 \ddot{a}^\pm _r\\
&\qquad  + \mathcal{O}\left( (\lambda - \lambda_r^\pm )^3\right), \\
a^\pm  (\lambda)  \mathop{\simeq}\limits_{\lambda \simeq -\lambda_r^{\pm ,*}} (\lambda + \lambda_r^{\pm ,*}) \dot{a}^\pm _{r+N_2}
&+ \frac{1}{2}  (\lambda + \lambda_r^{\pm ,*})^2 \ddot{a}^\pm _{r+N_2}
\\ &\qquad + \mathcal{O}\left( (\lambda + \lambda_r^{\pm ,*})^3\right) .
\end{split}\end{equation}
The reduction conditions (\ref{eq:ab*1}) impose constraints on the coefficients of eq. (\ref{eq:apm3}) which read:
\begin{equation}\label{eq:adotk}\begin{aligned}
\dot{a}_s^\pm &= - \dot{a}_s^{\pm , *}, &\quad \ddot{a}_s^\pm &=  \ddot{a}_s^{\pm , *}, \\
\dot{a}_r^+ &= - \dot{a}_{r+N_2}^{+ , *}, &\quad \ddot{a}_r^+ &=  \dot{a}_{r+N_2}^{+, *}, \\
\dot{a}_r^- &= - \dot{a}_{r+N_2}^{- , *}, &\quad \ddot{a}_r^+  &=  \ddot{a}_{r+N_2}^{- , *}.
\end{aligned}\end{equation}

The constants $\dot{a}_k^\pm$ and $\ddot{a}_k^\pm$ are not enough to characterize the discrete spectrum.
Indeed, the FAS $\chi^\pm (x,\lambda)$ become degenerate for $\lambda =\lambda_k^\pm$, which means that
if $\psi_k^\pm(x,t)= \psi^\pm(x,t,\lambda_k^\pm)$ and  $\phi_k^\pm(x,t)= \phi^\pm(x,t,\lambda_k^\pm)$ then
\begin{equation}\label{eq:xxx}\begin{split}
\phi_s^\pm (x,t) = \pm b_s^\pm(t) \psi^\pm_s(x,t), \\  \phi_r^\pm (x,t) &= \pm b_r^\pm (t)\psi^\pm_r(x,t),
\end{split}\end{equation}
and, as we shall see below, the constants $b_k^\pm(t)$ determine the angle variables.

The easiest way to find their properties   is to assume that the potential $q(x,t)$ is on a finite support, see Remark \ref{rem:1}. Then all the
matrix elements of the scattering matrix $a^\pm(\lambda)$ and  $b^\pm(\lambda)$ allow analytic extension
for any $\lambda \in \mathbb{C}$.
Consider now the  functions $b^\pm (\lambda,t)$ in the vicinity of $\lambda_k^\pm$:
\begin{equation}\label{eq:bpmk}\begin{split}
b^\pm(\lambda,t) \mathop{\simeq}\limits_{\lambda \simeq \lambda_k^\pm }  b^\pm_k(t)  +\mathcal{O}\left(   (\lambda -\lambda_k^\pm) \right), \\
b^\pm(\lambda,t) \mathop{\simeq}\limits_{\lambda \simeq \lambda_k^\mp }  \widetilde{b}^\pm_k(t) + \mathcal{O}\left(   (\lambda -\lambda_k^\mp) \right),
\end{split}\end{equation}
i.e., we have denoted by $b^\pm_k = b^\pm(\lambda_k^\pm)$ and $\widetilde{b}^\pm_k = b^\pm(\lambda_k^\mp)$.
In addition the relation $b^-(\lambda,t) = (b^+(-\lambda^*,t))^*$ gives:
\begin{equation}\label{eq:bpmk1}\begin{aligned}
 b^- (\lambda_s^+,t) &= (b^+ (-\lambda_s^{+,*},t))^*, \\  b^+ (\lambda_s^-,t) &= (b^+ (-\lambda_s^{-,*},t))^*,  \\
 b^- (\lambda_r^+,t) &= (b^+ (-\lambda_r^{+,*},t))^*, \\  b^+ (\lambda_r^-,t) &= (b^+ (-\lambda_r^{-,*},t))^*,
\end{aligned}\end{equation}
or with the notations just introduced:
\begin{equation}\label{eq:bpmk2}\begin{aligned}
b^\pm_s (t) &= \widetilde{b}^{\mp,*}_s(t), \\ b^\pm_r(t) &= \widetilde{b}^{\mp,*}_{r+N_2}(t), \qquad
\widetilde{b}^\pm_r(t) = b^{\mp,*}_{r+N_2}(t),
\end{aligned}\end{equation}
where $s=1,\dots, N_1$ and $r=N_1+1,\dots, N_1+N_2$.

Now we can consider  the fact that:
\begin{equation}\label{eq:det1}\begin{split}
\det T(\lambda,t) \equiv a^+(\lambda) a^-(\lambda) + b^+(\lambda,t) b^-(\lambda,t) =1,
\end{split}\end{equation}
which for finite support potentials holds true for all $\lambda$.
Putting $\lambda =\lambda_k^\pm$ into the left hand side of eq. (\ref{eq:det1}) we immediately find that:
\begin{equation}\label{eq:Bpmk}\begin{split}
b_k^+(t)\widetilde{b}^-_k (t)= 1, \qquad b_k^-(t)\widetilde{b}^+_k(t) = 1,
\end{split}\end{equation}
or
\begin{equation}\label{eq:Bpmk2}\begin{aligned}
 b^\pm_s(t) &= \frac{1}{b_s^{\pm ,*}(t)}. &\qquad \widetilde{b}^\pm_s(t) &= \frac{1}{\widetilde{b}_s^{\pm ,*}(t)},  \\
 b^\pm_r(t) &= \frac{1}{b_{r+N_2}^{\pm ,*}(t)}, &\qquad  \widetilde{b}^\pm_r (t)&= \frac{1}{\widetilde{b}_{r+N_2}^{\pm ,*}(t)}.
\end{aligned}\end{equation}

Now we have to take the limit when the support of the potential goes to infinity. In this case the functions $b^\pm(\lambda)$
cannot be extended off the real axis and  $b_k^\pm$ and $\widetilde{b}^\pm_k$ are just some constants. However the relations
(\ref{eq:bpmk2}), (\ref{eq:Bpmk}) and (\ref{eq:Bpmk2}) still remain valid.

Let us now analyze the effect of the nonlocal involution A) on the trace identities, see Subsection IV.3.
Our first remark is that the involution will not change the formal expressions of the integrals of motion $C_k$ in terms
of $q_\pm (x,t)$. Indeed, we just have to replace in them $q_+(x,t)$ by $u(x,t) $ and $q_-(x,t)$ by $u^*(-x,t) $.
However, the expressions for $C_k$ in terms of the scattering data (or the action variables) will change substantially in view of
Lemma \ref{lem:2}. Skipping the details we obtain.
\begin{equation}\label{eq:}\begin{split}
\mathcal{A}(\lambda) = i \sum_{p=1}^{\infty} \frac{C_p}{\lambda^p},
\end{split}\end{equation}
where
\begin{multline}\label{eq:C-2m}
C_{2m}\\
 = - i \int_{0}^{\infty} d\mu \; \mu^{2m-1} \eta_1(\mu) + \frac{i (-1)^m}{4^m 2m}
\sideset{}{^\pm }\sum_{s=1}^{N_1} \left(\pm (\eta_s^\pm)^{2m}  \right) \\
+\frac{i (-1)^m}{4^m 2m} \sideset{}{^\pm }\sum_{s=1}^{N_1} \left(\pm (\eta_r^\pm)^{2m}  \pm (\eta_r^{\pm,*})^{2m} \right) ,
\end{multline}
\begin{multline}\label{eq:C-2m-1}
C_{2m-1}\\
 = -  \int_{0}^{\infty} d\mu \; \mu^{2m-2} \eta_0(\mu) + \frac{ (-1)^m 2}{4^m (2m-1)} \sideset{}{^\pm }\sum_{s=1}^{N_1}  (\eta_s^\pm)^{2m-1}   \\
+\frac{ (-1)^m 2}{4^m (2m-1)} \sideset{}{^\pm }\sum_{r=1}^{N_1} \left( (\eta_r^\pm)^{2m-1}  +(\eta_r^{\pm,*})^{2m-1} \right) .
\end{multline}
Note that all `even' conserved quantities $C_{2m} $ are purely imaginary, while all `odd' ones $C_{2m-1}$ take real values.

\begin{theorem}\label{thm:1AA}
The set of the action-angle variables for the nonlocal NLS in the case of involution A) is given by:

\begin{subequations}\label{eq:AA3}
\begin{equation}\label{eq:AA3a}\begin{aligned}
\kappa (\lambda,t) &= \frac{1}{2} \ln \frac{b^+(\lambda,t)}{b^-(\lambda,t)} = -\kappa^*(-\lambda,t), \\
\eta(\lambda)& =-\frac{1}{\pi} \ln \left( a^+a^-(\lambda)\right) = \eta^*(-\lambda),
\end{aligned}\end{equation}
\begin{equation}\label{eq:AA3b}\begin{aligned}
\eta_s^\pm & = \mp 2i \lambda_s^\pm = 2n_s^\pm , &\quad \kappa_s^\pm (t)&=\pm i \arg b_s^\pm (t),
\end{aligned}\end{equation}
\begin{equation}\label{eq:AA3c}\begin{aligned}
\eta_r^\pm & = \mp 2i \lambda_r^\pm = \mp 2ip_r^\pm  +2s_r^\pm , \\
 \kappa_r^\pm (t) &=\pm \ln |b_r^\pm (t) | \pm i \arg b_r^\pm (t), \\
\eta_{r+N_2}^\pm & = \pm 2i \lambda_r^{\pm,*} = \pm 2ip_r^\pm  +2s_r^\pm , \\
 \kappa_{r+N_2}^\pm (t) &=\mp \ln |b_r^\pm (t)| \pm i \arg b_r^\pm(t) =-\kappa_r^{\pm,*}(t) ,
\end{aligned}\end{equation}
\end{subequations}
where $\lambda\in \mathbb{R}$ runs over the continuous spectrum of $L$, $s=1,\dots, N_1$ and $r=N_1+1, \dots, N_1+N_2$.

\end{theorem}

\begin{proof}
Let us now insert these expressions into the formula for $\Omega_0$ (\ref{eq:9.4}). In what follows we will split $\kappa(\lambda)$
and $\eta(\lambda)$ into real and imaginary parts:
\begin{equation}\label{eq:xxx1}\begin{split}
\kappa(\lambda ,t) & = \kappa_0(\lambda ,t)+i \kappa_1(\lambda ,t), \qquad \eta(\lambda) = \eta_0(\lambda) + i \eta_1(\lambda) ,
\end{split}\end{equation}
with
\begin{equation}\label{eq:xxx2}\begin{aligned}
 \kappa_0(\lambda ,t) &= \frac{1}{2} \ln \left| \frac{b^+(\lambda,t)}{b^-(\lambda,t)}\right| , \\
 \kappa_1(\lambda ,t) &= \frac{1}{2} \left( \arg (b^+(\lambda,t) - \arg (b^-(\lambda,t)\right), \\
 \eta_0(\lambda) &= -\frac{1}{\pi} \ln \left| a^+a^-(\lambda) \right| , \\
 \eta_1(\lambda) &=- \frac{1}{\pi} \left( \arg (a^+(\lambda,t) + \arg (a^-(\lambda,t)\right) .
\end{aligned}\end{equation}

 For the expansion of $\sigma_3 \delta q(x)$ we obtain:
\begin{multline}\label{eq:dqsb}
\sigma_3 \delta q(x) = i \int_{-\infty}^{\infty} d\lambda \; \left( \mathcal{Q}(x,\lambda) \delta \eta(\lambda) -  \mathcal{P}(x,\lambda) \delta \kappa(\lambda) \right) \\
+ i  \sideset{}{^\pm}\sum_{r=N_1+1}^{N_1+N_2} \left( \mathcal{Q}_r^\pm (x) \delta \eta_r^\pm -  \mathcal{P}_r^\pm (x) \delta \kappa_r^\pm \right. \\
\left. + \mathcal{Q}_{r+N_2}^\pm (x) \delta \eta_{r+N_2}^\pm -  \mathcal{P}_{r+N_2}^\pm (x) \delta \kappa_{r+N_2}^\pm \right)  \\
+ i  \sideset{}{^\pm}\sum_{s=1}^{N_1} \left( \mathcal{Q}_s^\pm (x) \delta \eta_s^\pm -  \mathcal{P}_s^\pm (x) \delta \kappa_s^\pm \right) .
\end{multline}

In order to express the canonical symplectic form  $\Omega_0$ in terms of the action-angle variables
it is enough to take the skew-scalar product of $\frac{i}{2} \sigma_3 \delta q$ with both sides of eq. (\ref{eq:dqsb}) and take into account the second line of
Table \ref{tab:2}. The result is
\begin{multline}\label{eq:Ome0}
\Omega_{(0)} \equiv \frac{i}{2} \biglb \sigma_3 \delta q \wedgecomma \sigma_3\delta q\bigrb \\
= i \int_{-\infty}^{\infty} d\lambda\; \delta\eta(\lambda) \wedge \delta \kappa(\lambda) + i \sideset{}{^\pm}\sum_{s=1}^{N_1}
\delta\eta_s^\pm \wedge  \delta\kappa_s^\pm \\
+ i \sideset{}{^\pm}\sum_{r=N_1+1}^{N_1+N_2} \left( \delta\eta_r^\pm \wedge  \delta\kappa_r^\pm
+\delta\eta_{r+N_1}^\pm \wedge  \delta\kappa_{r+N_1}^\pm \right) \\
= -2 \int_{0}^{\infty} d\lambda\; \left(\delta\eta_0(\lambda) \wedge \delta \kappa_1(\lambda) +\delta\eta_1(\lambda) \wedge \delta \kappa_0(\lambda)\right) \\
- \sideset{}{^\pm}\sum_{s=1}^{N_1} \delta\eta_s^\pm \wedge  \delta\arg b_s^\pm \\
+ i \sideset{}{^\pm}\sum_{r=N_1+1}^{N_1+N_2} \left( \delta\eta_r^\pm \wedge  \delta\kappa_r^\pm
- \delta\eta_{r}^{\pm*} \wedge  \delta \kappa_{r}^{\pm,*} \right) ,
\end{multline}
where we made use of eqs. (\ref{eq:AA3}).

Now the proof follows from the following easy observations:\\
i) the equation:
\begin{equation}\label{eq:OmedH}\begin{split}
\Omega_{(0)} (i \sigma_3 q_t, \cdot) = \delta H_{(0)}
\end{split}\end{equation}
where $H_{(0)}= -2C_3$ reproduces the time-dependence of the scattering data, see eq. (\ref{eq:7.12c});\\
ii) The symplectic form $\Omega_{(0)}$ is canonical with respect to the set of action-angle variables
$\{ \eta(\lambda), \kappa(\lambda)\}$; \\
$H_{\rm nlNLS} $ depends only on the action-type variables $\eta(\lambda)$. \\
iv) The set of symplectic basis is complete.

These facts can be viewed  as natural generalization of the notion of AA variables to the
infinite-dimensional case.
\end{proof}

The `even' hierarchy of symplectic forms is:
\begin{multline}\label{eq:OmeE}
\Omega_{(2m)} \equiv \frac{i}{2} \biglb \sigma_3 \delta q \wedgecomma \Lambda^{2m} \sigma_3\delta q\bigrb \\
= 2 \int_{0}^{\infty} d\lambda\;  \lambda^{2m} \left(\delta\eta_0(\lambda) \wedge \delta \kappa_1(\lambda) +\delta\eta_1(\lambda) \wedge \delta \kappa_0(\lambda) \right)  \\
+ c_{1,m}\sideset{}{^\pm} \sum_{s=1}^{N_1} \left(\pm \delta (\eta_s^\pm)^{2m+1} \wedge  \delta\arg b_s^\pm  \right) \\
- i c_{1,m}  \sideset{}{^\pm}\sum_{r=N_1+1}^{N_1+N_2} \left( \delta (\eta_r^\pm)^{2m+1} \wedge  \delta\kappa_r^\pm
- \delta (\eta_{r}^{\pm*})^{2m+1} \wedge  \delta \kappa_{r}^{\pm,*} \right),
\end{multline}
where $c_{1,m} = (-1)^m 4^{-m} (2m+1)^{-1}$.

We end with the examples of the simplest NLEE satisfying the nonlocal reduction.
Their dispersion laws  are given by even functions of $\lambda$ with real coefficients.
\begin{eqnarray*}\label{eq:disp-red2}
& f_{\rm nl NLS}(\lambda ) = -\lambda ^2, \qquad f_{\rm nl NLS4}(\lambda ) = -4 \lambda ^4.
\end{eqnarray*}

\subsection{The nonlocal involution B: $q_-(-x,-t)=(q_+)^*(x,t) =v(x,t)$ }\label{ssec:10.2}

As mentioned above, the involution A) is compatible with the NLS equation and also with all
higher NLS equations whose dispersion law $f(\lambda)$ is an even function of $\lambda$.
However the involution A) is not compatible with the cmKdV equation
\begin{equation}\label{eq:cmKdV}\begin{split}
\frac{\partial v}{ \partial t} + \frac{\partial ^3 v}{ \partial x^3 } + 6 v(x,t) v^*(-x,-t) \frac{\partial v}{ \partial x }=0
\end{split}\end{equation}
 and all higher cmKdV-like equations whose dispersion law $f(\lambda)$ is an odd function of $\lambda$. These equations allow
reduction B)
\begin{equation}\label{eq:B}\begin{split}
q_-(-x,-t)=(q_+)^*(x,t) =v(x,t) ,
\end{split}\end{equation}
which requires simultaneous change of the sign of both $x$ and $t$. This involution leads to:
\begin{equation}\label{eq:Ue*2}\begin{split}
U(x,t,\lambda) = \sigma_1 \left( U(-x,-t,-\lambda^*)\right)^* \sigma_1 ^{-1},
\quad \sigma_1  = \left(\begin{array}{cc} 0 &  1 \\ 1 & 0  \end{array}\right).
\end{split}\end{equation}

Again  the involution on the potential $q(x,t)$ can be extended to an involutive automorphism
of the affine Lie algebra in which $U(x,t,\lambda)$ and $V(x,t,\lambda)$ are taking values.
All fundamental solutions of $L$ are also restricted by it; e.g., the Jost solutions must satisfy \cite{AblMus,AblMus2}:
\begin{equation}\label{eq:Jo*3}\begin{split}
\psi^-(x,t,\lambda) &= \sigma_1 ^{-1} \left( \phi^- (-x,-t,-\lambda^*)\right)^*, \\
\psi^+(x,t,\lambda) &= \sigma_1  \left( \phi^+ (-x,-t,-\lambda^*)\right)^* ,
\end{split}\end{equation}
or in compact form:
\begin{equation}\label{eq:Jo*4}\begin{split}
\psi(x,t,\lambda) = \sigma_1  \left( \phi (-x,-t,-\lambda^*)\right)^* \sigma_1^{-1} .
\end{split}\end{equation}
From eq. (\ref{eq:AKNS2})  we derive the following constraint for the scattering matrix:
\begin{equation}\label{eq:T*1}\begin{split}
T(\lambda,t) = \sigma_1  \left(  \widehat{T}(-\lambda^*,-t) \right)^* \sigma_1 ^{-1};
\end{split}\end{equation}
similarly from eq. (\ref{eq:Fas0}) for the FAS we find:
\begin{equation}\label{eq:chiR1'}\begin{split}
\chi^+(x,t,\lambda)  &= \sigma_1  \left( \chi^+ (-x,-t,-\lambda^*)\right)^* \sigma_1 ^{-1}  , \\
\chi^-(x,t,\lambda)  &= \sigma_1  \left( \chi^- (-x,-t,-\lambda^*)\right)^* \sigma_1 ^{-1}  .
\end{split}\end{equation}

Under this involution  the dispersion law of the NLEE must satisfy:
\begin{equation}\label{eq:f*m}\begin{split}
f(\lambda) &= -(f(-\lambda^*))^*, \quad \mbox{i.e.} \quad f(\lambda) = \sum_{k>0}^{}  f_{2k+1} \lambda^{2k+1}  ,
\end{split}\end{equation}
which means that $f(\lambda)$ must be an odd polynomial of $\lambda$ with real coefficients. In particular, the involution
B) is compatible with the mKdV eq., because its dispersion law is $f_{\rm mKdV}(\lambda)=-4\lambda^3$.
The eq. (\ref{eq:T*1}) in components reads:
\begin{equation}\label{eq:ab*1B}\begin{split}
a^+(\lambda) &= (a^+(-\lambda^*))^*, \qquad a^-(\lambda) = (a^-(-\lambda^*))^*, \\
b^-(t,\lambda) &=   (b^+(-t,-\lambda^*))^*.
\end{split}\end{equation}

Note that  for  $a^+(\lambda)$ and $a^-(\lambda)$ the constraints are the same as for the involution A) -- indeed
$a^\pm (\lambda)$ are time-independent. That is why many of the considerations in the previous Subsection are
directly applicable also for involution B). In particular this involves Corollary \ref{cor:4}, Remark \ref{rem:5} and
Lemma \ref{lem:2}. Of course the restrictions on $b^\pm_k(t)$ and $\tilde{ b}^\pm_k(t)$ must be modified into:

\begin{equation}\label{eq:bpmk2B}\begin{aligned}
 b^- (\lambda_s^+,t) &= (b^+ (-\lambda_s^{+,*},-t))^*, \\  b^+ (\lambda_s^-,t) &= (b^+ (-\lambda_s^{-,*},-t))^*,  \\
 b^- (\lambda_r^+,t) &= (b^+ (-\lambda_r^{+,*},-t))^*, \\  b^+ (\lambda_r^-,t) &= (b^+ (-\lambda_r^{-,*},-t))^*,
\end{aligned}\end{equation}
or with the above notations:
\begin{equation}\label{eq:bpmk2'B}\begin{aligned}
b^\pm_s (t) &= \widetilde{b}^{\mp,*}_s(-t), \\ b^\pm_r(t) &= \widetilde{b}^{\mp,*}_{r+N_2}(-t), \qquad
\widetilde{b}^\pm_r(t) = b^{\mp,*}_{r+N_2}(-t),
\end{aligned}\end{equation}
where $s=1,\dots, N_1$ and $r=N_1+1,\dots, N_1+N_2$.
In addition, from the fact that $\det T(\lambda,t)=1$ we get: that:
\begin{equation}\label{eq:Bpmk'}\begin{split}
b_k^+(t)\widetilde{b}^-_k (t)= 1, \qquad b_k^-(t)\widetilde{b}^+_k(t) = 1,
\end{split}\end{equation}
and
\begin{equation}\label{eq:Bpmk2B}\begin{aligned}
 b^\pm_s(t) &= \frac{1}{b_s^{\pm ,*}(-t)}. &\qquad \widetilde{b}^\pm_s(t) &= \frac{1}{\widetilde{b}_s^{\pm ,*}(-t)},  \\
 b^\pm_r(t) &= \frac{1}{b_{r+N_2}^{\pm ,*}(-t)}, &\qquad  \widetilde{b}^\pm_r (t)&= \frac{1}{\widetilde{b}_{r+N_2}^{\pm ,*}(-t)}.
\end{aligned}\end{equation}

\begin{theorem}\label{thm:4}
In the case of involution B) the set of action-angle variables for the nonlocal cmKdV-like equations is given by:
\begin{subequations}\label{eq:AA4}
\begin{equation}\label{eq:AA4a}\begin{aligned}
\kappa (\lambda,t) &= \frac{1}{2} \ln \frac{b^+(\lambda,t)}{b^-(\lambda,t)} = -\kappa^*(-\lambda,-t), \\
\eta(\lambda)& =-\frac{1}{\pi} \ln \left( a^+a^-(\lambda)\right) = \eta^*(-\lambda),
\end{aligned}\end{equation}
\begin{equation}\label{eq:AA4b}\begin{aligned}
\eta_s^\pm & = \mp 2i \lambda_s^\pm = 2n_s^\pm , &\quad \kappa_s^\pm (t) &=\pm i \arg b_s^\pm (t),
\end{aligned}\end{equation}
\begin{equation}\label{eq:AA4c}\begin{aligned}
\eta_r^\pm & = \mp 2i \lambda_r^\pm = \mp 2ip_r^\pm  +2s_r^\pm , \\
 \kappa_r^\pm (t) &=\pm \ln |b_r^\pm (t)| \pm i \arg b_r^\pm (t), \\
\eta_{r+N_2}^\pm & = \pm 2i \lambda_r^{\pm,*} = \pm 2ip_r^\pm  +2s_r^\pm , \\
 \kappa_{r+N_2}^\pm (t) &=\mp \ln |b_r^\pm (t)| \pm i \arg b_r^\pm(t) =-\kappa_r^{\pm,*}(-t) ,
\end{aligned}\end{equation}
\end{subequations}
where $\lambda\in \mathbb{R}$ runs over the continuous spectrum of $L$, $s=1,\dots, N_1$ and $r=N_1+1, \dots, N_1+N_2$.

\end{theorem}

\begin{proof}
The proof is similar to the one of Theorem \ref{thm:1AA}. The only difference is that now we have to
apply involution B) on the symplectic forms $\Omega_{(m)}$ and on the scattering data. The
expressions for the action variables are $t$-independent, so they are the same as in (\ref{eq:xxx2}). 
Again we will split $\kappa(\lambda)$ and $\eta(\lambda)$ into real and imaginary parts:
\begin{equation}\label{eq:xxx2B}\begin{split}
\kappa(\lambda) & = \kappa_0(\lambda)+i \kappa_1(\lambda), \qquad \eta(\lambda) = \eta_0(\lambda) + i \eta_1(\lambda) ,
\end{split}\end{equation}
with
\begin{equation}\label{eq:xxx3}\begin{aligned}
 \kappa_0(\lambda) &= \frac{1}{2} \ln \left| \frac{b^+(\lambda,t)}{b^-(\lambda,t)}\right|  = \frac{1}{2} \ln \left| \frac{b^+(\lambda,t)}{b^+(-\lambda,-t)}\right|, \\
 \kappa_1(\lambda) &= \frac{1}{2} \left( \arg (b^+(\lambda,t) - \arg (b^-(\lambda,t)\right), \\
 &= \frac{1}{2} \left( \arg (b^+(\lambda,t) + \arg (b^+(-\lambda^*,-t)\right).
\end{aligned}\end{equation}

Similarly as above we  express the canonical symplectic form  $\Omega_0$ in terms of the action-angle variables as follows:
\begin{multline}\label{eq:Ome0B}
\Omega_{(0)} \equiv \frac{i}{2} \biglb \sigma_3 \delta q \wedgecomma \sigma_3\delta q\bigrb \\
= -2 \int_{0}^{\infty} d\lambda\; \left(\delta\eta_0(\lambda) \wedge \delta \kappa_1(\lambda) +\delta\eta_1(\lambda) \wedge \delta \kappa_0(\lambda)\right) \\
- \sideset{}{^\pm}\sum_{s=1}^{N_1} \delta\eta_s^\pm \wedge  \delta\arg b_s^\pm \\
+ i \sideset{}{^\pm}\sum_{r=N_1+1}^{N_1+N_2} \left( \delta\eta_r^\pm \wedge  \delta\kappa_r^\pm
- \delta\eta_{r}^{\pm*} \wedge  \delta \kappa_{r}^{\pm,*} \right) ,
\end{multline}
where we made use of eqs. (\ref{eq:AA3}).

Now the complete integrability follows from the  easy observations:\\
i) the equation:
\begin{equation}\label{eq:OmedH2}\begin{split}
\Omega_{(0)} (i \sigma_3 q_t, \cdot) = \delta H_{(0)}
\end{split}\end{equation}
where $H_{(0)}= 16 iC_4$ coincides with the cmKdV; \\
ii) The symplectic form $\Omega_{(0)}$ is canonical with respect to the set of action-angle variables
$\{ \eta(\lambda), \kappa(\lambda)\}$; \\
$H_{(0)} $ depends only on the action-type variables $\eta(\lambda)$. \\
iv) The set of symplectic basis is complete.

\end{proof}
We end by expressing also the `odd' hierarchy of symplectic forms in terms of the action-angle variables:
\begin{multline}\label{eq:OmeE2}
\Omega_{(2m-1)} \equiv \frac{i}{2} \biglb \sigma_3 \delta q \wedgecomma \Lambda^{2m-1} \sigma_3\delta q\bigrb \\
= 2i \int_{0}^{\infty} d\lambda\;  \lambda^{2m-1} \left(\delta\kappa_0(\lambda) \wedge \delta \eta_0(\lambda) -\delta\kappa_1(\lambda) \wedge \delta \eta_1(\lambda) \right)  \\
+ i c_{m}\sideset{}{^\pm} \sum_{s=1}^{N_1} \left(\pm \delta (\arg b_s^\pm) \wedge  \delta (\eta_s^\pm )^{2m} \right) \\
+ 2i c_{m}  \sideset{}{^\pm}\sum_{r=N_1+1}^{N_1+N_2} \left( \delta (\ln |b_r^\pm|) \wedge  \delta y_{1,m}^\pm
+ \delta (\arg b_r^\pm ) \wedge  \delta y_{0,m}^\pm  \right),
\end{multline}
where
\[ c_{1,m} =  \frac{(-1)^m }{ 4^{m} m}, \qquad y_{0,m}^\pm +i y_{1,m}^\pm = (2 s_r^\pm \mp 2ip_r^\pm)^{2m}.\]

\section{Discussion and conclusions}

We have proved that the nonlocal equations from the hierarchy of the nonlocal NLS equations
are infinite-dimensional, completely integrable systems. The proof is based on the so-called completeness
relation for the squared solutions of the Lax operator and the symplectic basis, which allows one to map
directly the variation of the potential $\sigma_3\delta q(x)$ into the variations of the action-angle variables.

The spectral properties of the Lax operator (see Corollary \ref{cor:4} and Remark \ref{rem:5})  allow us to conclude that the hierarchy of non-local
equations has two types of soliton solutions, see \cite{AblMus2,Val}. Their properties and interactions will be analyzed elsewhere.

The corresponding action-angle variables are parametrized by the scattering data of $L$. More specifically
we have complex-valued action variables $\eta(\lambda)$ and $\kappa(\lambda)$
related to the continuous spectrum of $L$ (i.e., for $\lambda\in \mathbb{R}$) and two finite sets of action-angle
variables: $\eta_s^\pm$ and $\kappa_s^\pm$ related to the solitons of first type ($s=1,\dots, N_1$) and
$\eta_r^\pm$, $\eta_r^{\pm,*}$ and $\kappa_r^\pm$, $\kappa_r^{\pm,*}$ related to the solitons of second type ($r=N_1+1,\dots, N_1+N_2$).

These action-angle variables are universal in the sense that they render canonical each of the symplectic forms in the hierarchy 
$\Omega_{(m)}$. The action-angle variables here are understood in the sense that the action variables are $t$-independent,
while the angle variables depend linearly on $t$. The requirement that the angle variables are in the interval $[0, 2\pi]$
cannot be ensured for all $\kappa(\lambda)$ and $\kappa_k^\pm$. 

\section*{Acknowledgements}
This work was supported in part by the U.S. Department of Energy. We thank an anonymous referee for
careful reading of the manuscript and for useful suggestions.

\end{document}